\documentclass[12pt]{article}
\usepackage[utf8]{inputenc}

\usepackage{geometry,amsmath,mathtools,blkarray,amssymb,amsthm,setspace,xcolor,bbm,tikz-cd,xparse,amsfonts,authblk,graphics,datetime,thmtools,bm}

\usepackage[normalem]{ulem}

\usepackage{caption,subcaption}

\usepackage{tikz,tikz-3dplot}    
\usetikzlibrary{shapes,decorations,arrows,calc,arrows.meta,fit,positioning}

\tikzset{
    -Latex,auto,node distance =1 cm and 1 cm,semithick,
    state/.style ={ellipse, draw, minimum width = 0.7 cm},
    point/.style = {circle, draw, inner sep=0.04cm,fill,node contents={}},
    bidirected/.style={Latex-Latex},
    el/.style = {inner sep=2pt, align=left, sloped}
}

\usepackage[shortlabels]{enumitem}
\usepackage[colorlinks,linkcolor=blue,citecolor=blue,urlcolor=blue,bookmarks=false,hypertexnames=true]{hyperref}
\usepackage[backend=biber, style=bwl-FU, sortcites=false, maxcitenames=2, mincitenames=1, maxbibnames=8, uniquelist=false]{biblatex}

\allowdisplaybreaks
\bibliography{ref.bib}

\providecommand{\keywords}[1]
{
  \small	
  \textbf{{Keywords---}} #1
}

\theoremstyle{plain}
\newtheorem{theorem}{Theorem}[section]
\newtheorem{proposition}{Proposition}[section]
\newtheorem{lemma}{Lemma}[section]
\newtheorem{corollary}{Corollary}[section]

\theoremstyle{definition}

\newtheorem{definition}{Definition}[section]
\newtheorem{example}{Example}
\newtheorem{assumption}{Assumption}[section]

\newtheorem{remark}{Remark}[section]

\renewcommand\thmcontinues[1]{Continued}

\newcommand{\indep}{{\perp \!\!\! \perp}}
\newcommand{\tr}{\mathrm{tr}}
\newcommand{\norm}[1]{\left\|#1\right\|}
\newcommand{\ip}[2]{\left<#1, #2\right>}
\newcommand{\argmin}{\mathop{\mathrm{argmin}}}
\newcommand{\argmax}{\mathop{\mathrm{argmax}}}

\newcommand{\memo}[2]{#2}
\renewcommand{\div}{\mathrm{div}}
\newcommand{\id}{\mathrm{id}}
\newcommand{\Int}{\hspace{0.1em}\mathrm{Int}\hspace{0.1em}}
\newcommand{\Cl}{\mathrm{Cl}}
\newcommand{\cof}{\hspace{0.1em}\mathrm{cof}\hspace{0.1em}}
\newcommand{\supp}{\hspace{0.1em} \mathrm{supp} \hspace{0.1em}}
\newcommand{\indic}{\bm{1}}

\NewDocumentCommand{\defmathletter}{m}{%
    \expandafter\newcommand\csname b#1\endcsname{\mathbb{#1}}%
    \expandafter\newcommand\csname c#1\endcsname{\mathcal{#1}}%
}
\NewDocumentCommand{\defmathletters}{>{\SplitList{,}}m}{\ProcessList{#1}{\defmathletter}}
\defmathletters{A,B,C,D,E,F,G,H,I,J,K,L,M,N,O,P,Q,R,S,T,U,V,X,Y,Z}

\NewDocumentCommand{\defvector}{m}{%
    \expandafter\newcommand\csname v#1\endcsname{\mathbf{#1}}%
}
\NewDocumentCommand{\defvectors}{>{\SplitList{,}}m}{\ProcessList{#1}{\defvector}}
\defvectors{A,B,C,D,E,F,G,H,I,J,K,L,M,N,O,P,Q,R,S,T,U,V,X,Y,Z}
\defvectors{a,b,c,d,e,f,g,h,i,j,k,l,m,n,o,p,q,r,s,t,u,v,x,y,z}

\newdateformat{monthyeardate}{\monthname[\THEMONTH] \THEYEAR}

\title{Identification and Estimation in a Nonseparable Instrumental Variable Model with Multivariate Outcomes}
\author{Haruki Kono \thanks{Email: h.kono@r.hit-u.ac.jp. I am grateful to Alberto Abadie, Isaiah Andrews, Kengo Kato, Dean Li, Anna Mikusheva, Whitney Newey, Ashesh Rambachan, Kazuatsu Shimizu, and participants in MIT econometrics lunch seminar and North American Summer Meeting of Econometric Society for their feedback and helpful discussions. I acknowledge financial support from the Funai Overseas Scholarship and the Jerry A. Hausman Fellowship.}}
\affil{Hitotsubashi University}
\date{\monthyeardate\today}

\begin{document}

\maketitle

\begin{abstract}

This paper proposes a novel nonseparable instrumental variable model to address endogeneity in settings with multivariate outcomes.
While traditional approaches typically estimate separate one-dimensional quantile models for each outcome, the inherent correlation among different dimensions frequently invalidates their fundamental assumptions. 
To overcome this limitation, our framework jointly models the distribution of multiple outcomes. 
By assuming that the structural function is cyclically monotone, the model accommodates flexible dependence across outcome dimensions. 
We establish the local identification of these structural functions in the presence of discrete treatments and instruments under a completeness-type condition, which reduces to the well-known full-rank condition in one dimension.
Furthermore, we introduce a computationally tractable estimator and prove its uniform consistency. 
Simulation studies demonstrate the estimator's stable performance in finite samples.

\bigskip

\noindent
\keywords{IVQR, local identification, multivariate quantiles, optimal transport}
\end{abstract}

\section{Introduction} \label{ivmqr:sec:introduction}

The instrumental variable quantile regression (IVQR) model introduced by \cite{chernozhukov2005iv} has become a widely used tool for estimating quantiles of potential outcomes in the presence of endogeneity.
See, for example, \cite{chernozhukov2004effects} and \cite{autor2017effect}.
While the IVQR framework assumes a one-dimensional outcome variable, researchers are often interested in settings with multiple outcomes.
For instance, \cite{chernozhukov2004effects} examine the effect of 401(k) participation on several wealth measures, including total wealth and financial assets.
In such cases, it is common to estimate separate IVQR models for each outcome dimension.
However, this practice raises several conceptual and empirical concerns.

First, running separate models ignores the correlation structure among outcome variables and is therefore silent about their joint distribution.
In the 401(k) example, this separate approach cannot capture the share of total wealth accounted for by financial assets.

More importantly, estimating separate quantile models makes rank similarity, a crucial assumption of \cite{chernozhukov2005iv}, unrealistic.
To apply the IVQR model in the 401(k) example, one must assume that individuals' holdings of financial assets depend solely on their preferences for those assets.
In reality, individuals allocate wealth across multiple assets simultaneously, and their preferences over financial assets are inherently linked to preferences over other asset categories.

To address these limitations, we propose a new nonseparable IV model that accommodates multidimensional outcomes that may be correlated with each other.
Specifically, we consider a multivariate extension of the potential outcome framework in \cite{chernozhukov2005iv}.
For a treatment $D = d$, let $Y_d$ denote the corresponding $p$-dimensional potential outcome.
Conditional on covariates, we assume it can be represented as $Y_d = q_d^\ast (U_d)$, where $q_d^\ast$ is a structural function and $U_d$ is a $p$-dimensional rank variable.
\memo{red}{The rank variable $U_d$ is a random vector that captures unobserved heterogeneity among observationally identical individuals, and can be interpreted as reflecting latent individual characteristics.}

Existing studies of endogenous quantile models impose various monotonicity restrictions on the structural function for identification.
In most studies on nonseparable models with one-dimensional outcomes (e.g., \cite{matzkin2003nonparametric}, \cite{chernozhukov2005iv}, \cite{torgovitsky2015identification}, \cite{d2015identification}), the structural function is assumed to be the quantile function of the potential outcome.
Extending their assumption, we require that the structural function is the gradient of a convex function, and in particular, that it is cyclically monotone.
This restriction implies that an increase in the $i$th component of $U_d$ increases the $i$th component of $Y_d$, while also allowing cross-dimensional effects: changes in one component of $U_d$ can affect other components of $Y_d$ positively or negatively.
Thus, our framework can capture dependence across outcome dimensions, which is a feature that the standard one-dimensional approach cannot accommodate.
Section \ref{ivmqr:sec:model} presents two examples of structural functions satisfying cyclical monotonicity.

Gradients of convex functions not only generalize the notion of monotonicity to multivariate outcomes, but also play a central role in optimal transport theory.
Recent developments in statistics have shown that such maps possess many desirable properties that justify viewing them as multivariate analogues of quantile functions (see, e.g., \cite{ekeland2012comonotonic}, \cite{chernozhukov2017monge}, \cite{hallin2021distribution}, \cite{ghosal2022multivariate}).
This perspective motivates us to refer to our framework as an \textit{instrumental variable multivariate quantile regression model}.

Structural functions observationally equivalent to the truth are characterized by solutions to a system of infinite-dimensional equations in terms of identified objects.
Based on this characterization, we provide a local identification result of the structural functions.

A key condition for identification is a version of the completeness condition proposed by \cite{newey2003instrumental}.
This generalizes the identification condition of the standard IVQR model.
Indeed, when the outcome is one-dimensional, our completeness condition is equivalent to \citeauthor{chernozhukov2005iv}'s full-rank condition of a matrix consisting of the conditional distribution of the outcome and treatment given the IV.

In higher dimensions, the completeness condition becomes more demanding, reflecting the cost of accounting for multiple outcomes.
To address this, we provide several readily verifiable sufficient conditions.
Specifically, in settings with binary treatments and IVs, the existence of an IV state that predicts the treatment assignment perfectly implies the completeness.
This sufficient condition is satisfied, for instance, when the IV represents a randomized treatment offer, as in \cite{abadie2002instrumental}.
Remarkably, the availability of such an IV ensures that the structural function is identified globally, rather than just locally.

Building upon these identification results, we propose a computationally tractable estimator for the multivariate structural functions. 
The core idea of our estimation procedure is to minimize an empirical loss function, which measures how far a given pair of conditional latent distributions is from being consistent with the conditional distributions of the observed outcomes within the proposed IV model.

By minimizing this loss subject to a restriction implied by IV independence, we estimate the conditional latent distributions. 
Using these estimates, we then recover the structural functions. 
We establish the uniform consistency of this proposed estimator.

Because our theoretical results guarantee only local identification, the loss function might admit multiple minimizers, which could lead to instability in the estimator. 
We evaluate this concern through Monte Carlo simulations using a simple data-generating process. 
Despite the potential multiplicity of minimizers and the non-convexity of the loss function, our simulations reveal that the estimator consistently recovers the true structural functions regardless of the initialization.
These results provide strong computational evidence supporting both the global identification of our model under additional assumptions and the practical reliability of the estimation procedure.

\bigskip

\noindent
\textbf{Related literature.}
Our model builds on the IVQR framework proposed and developed by \cite{chernozhukov2005iv}, \cite{chernozhukov2006instrumental}, \cite{chernozhukov2007instrumental}, and \cite{chernozhukov2013quantile}.
Unlike these studies, we allow for multidimensional outcomes and multidimensional unobserved heterogeneity.

Several related papers consider endogeneity in models with a single outcome variable.
\cite{abadie2002instrumental} analyze quantile treatment effects using the local average treatment effect framework of \cite{imbens1994identification}.
Their approach identifies treatment effects for compliers but is limited to binary treatments.
Extending \cite{chesher2003identification}, \cite{imbens2009identification} adopt a control function approach under the assumption that the selection equation is monotone and the instrumental variable is independent of the unobserved disturbance.
Nonparametric identification in triangular models with discrete instruments is studied by \cite{torgovitsky2015identification} and \cite{d2015identification}, and further extended to multivariate settings by \cite{gunsilius2023condition}.
\memo{red}{
\cite{matzkin2008identification} investigates nonparametric identification in simultaneous equation models.
Her identification argument exploits a continuously varying endogenous variable and the invertibility of the structural system.
In contrast, we focus on discrete treatments and instruments and impose no monotonicity or invertibility restriction on the treatment-selection equation. 
Identification instead follows from rank similarity and the conditional restrictions generated by the discrete instrument.
}
Within the same simultaneous equation framework, \cite{blundell2017individual} propose a nonseparable model that exploits proxy variables for unobserved heterogeneity.

Our key restriction on the structural function requires it to be a multivariate quantile function, as developed by \cite{ekeland2012comonotonic}, \cite{chernozhukov2017monge}, \cite{hallin2021distribution}, and \cite{ghosal2022multivariate}.
Multivariate quantile functions are grounded in optimal transport theory (\cite{villani2003topics}, \cite{villani2009optimal}), which provides a natural framework for extending quantiles to random vectors.
\cite{galichon2018optimal} offers a comprehensive overview of economic applications of optimal transport.
In this paper, we draw on results from this literature to establish the identification of our model.

\bigskip

\noindent
\textbf{Organization:}
This paper is organized as follows.
Section \ref{ivmqr:sec:model} introduces the IV multivariate quantile regression model.
Section \ref{ivmqr:sec:identification} formally states the local identification result for the model and discusses the relationship to existing results.
Section \ref{ivmqr:sec:estimation} proposes an estimator of the structural function and shows its uniform consistency.
Section \ref{ivmqr:sec:conclusion} concludes.
All proofs and mathematical preliminaries are provided in the Appendix.

\section{IV multivariate quantile regression model} \label{ivmqr:sec:model}

\subsection{Model}

Let $Y_d$ be a $p$-dimensional potential outcome vector, $D \in \cD$ an endogenous variable (treatment), $X \in \cX$ an exogenous covariate vector, and $Z \in \cZ$ an IV.
For now, the treatment and IV can be either discrete or continuous, although we will focus on the discrete cases later.
\memo{red}{Let $Y \coloneqq Y_D$ be the observed outcome.}
Let $\cU \subset \bR^p$ be a compact convex set\footnote{The compactness and convexity of $\mathcal U$ are maintained regularity conditions rather than without-loss-of-generality normalizations. They are satisfied by standard choices of the reference distribution, such as the uniform distribution on $[0,1]^p$, and simplify the convex-analytic and uniform arguments used below. Allowing reference distributions with unbounded support would require additional tail and regularity conditions and is not considered in this paper.} with a piecewise $C^1$ boundary $\partial \cU.$\footnote{See Appendix \ref{ivmqr:app:preliminaries} for the definition. If $\cU$ has a smooth boundary or is a rectangle, it has a piecewise $C^1$ boundary.}
Let $\mu$ be a reference probability measure whose support is $\cU.$
Assume $\mu$ is absolutely continuous with respect to the $p$-dimensional Lebesgue measure.

Our model extends \cite{chernozhukov2005iv} so that it accounts for multidimensional outcome vectors.
\memo{red}{
Specifically, we assume that we have iid observations $(Y, D, Z, X)$ generated by the following structural model
\begin{equation} \label{ivmqr:eq:structural-equation}
    \begin{cases}
        Y = Y_D \\ 
        Y_d = q_d^\ast (U_d, X) \text{ for each $d \in \cD$} \\
        D = \delta (Z, X, \nu)
    \end{cases}
\end{equation}
where $U_d$ is the $p$-dimensional latent feature vector generating the potential outcome $Y_d,$ and $\nu$ is the latent variable (with arbitrary dimension) generating the treatment $D$ that may depend on $(U_d)_{d \in \cD}.$
Assumptions \ref{ivmqr:ass:ivqr-A1} and \ref{ivmqr:ass:ivqr-A2} below imply that \(U_d\) is independent of \((X,Z)\) for every \(d\in\mathcal D\).

This model can be understood in the 401(k) example as follows.
Individuals determine whether they participate in 401(k) or not ($D \in \{0, 1\}$) based on the eligibility $Z,$ observed covariates $X$ and the selection-relevant unobserved heterogeneity $\nu.$
The counterfactual portfolio vector $Y_d$ for $D = d$ is determined by applying $q_d^\ast$ to the (potential) unobserved heterogeneity $U_d$ and the covariates $X.$
We observe only the realized portfolio $Y = Y_D.$
Because the selection heterogeneity $\nu$ may depend on $(U_d)_{d\in\mathcal D},$ treatment selection may be endogenous. 
The precise exogeneity and rank-similarity restrictions are stated in Assumption \ref{ivmqr:ass:ivqr} below.
}

\memo{red}{
The standard IVQR model of \cite{chernozhukov2005iv} assumes the monotonicity of the structural function with respect to the latent variable.
In their work, potential outcomes, which are one-dimensional, are assumed to vary with scalar latent variables monotonically, i.e., the structural function $q_d^\ast (\cdot, x)$ is increasing.
If we normalize the distribution of the latent variable to $U [0, 1],$ the monotonicity implies that $q_d^\ast (u, x)$ gives the $u$-quantile of the distribution of $Y_d$ conditional on $X = x.$
Typically, we interpret it as the potential outcome level of an individual whose latent rank is $u.$

Monotonicity does not extend straightforwardly to multivariate functions. 
Although this is not the only possible extension, we impose the following notion on our structural functions.
\begin{definition} \label{ivmqr:def:cyclical-monotonicity}
A function $q : \cU \to \bR^p$ is \textit{cyclically monotone} if for any $k \geq 1$ and $u^1, \dots, u^{k + 1} \in \cU$ with $u^1 = u^{k + 1},$ it holds that
\begin{align} \label{ivmqr:eq:CM}
    \sum_{i = 1}^{k}
    \left(
        u^{i + 1}
    \right)^\prime
    \left(
        q (u^{i + 1}) 
        - 
        q (u^i)
    \right)
    \geq
    0
    .
\end{align}
It is \textit{strictly cyclically monotone} if it is cyclically monotone, and inequality (\ref{ivmqr:eq:CM}) is strict unless $u^1 = \dots = u^k.$
\end{definition}
When the function is one-dimensional, the cyclical monotonicity coincides with the non-decreasingness in the usual sense.
This notion has been used in the economic literature on mechanism design (\cite{rochet1987necessary}) and econometrics (\cite{shi2018estimating}) among others.

To see why cyclical monotonicity is a reasonable notion of multivariate monotonicity, we recall (a version of) Rockafellar's theorem (\cite{rockafellar1966characterization}, \cite{rockafellar1970maximal}) in convex analysis.
\begin{theorem}[Rockafellar] \label{ivmqr:thm:rockafellar}
Let $q: U \to \mathbb{R}^p$ be a continuously differentiable vector-valued function. 
The following two statements are equivalent:
\begin{enumerate}
    \item $q$ is (resp., strictly) cyclically monotone.
    \item There exists a differentiable (resp., strictly) convex function $\varphi: U \to \mathbb{R}$ such that $\nabla \varphi = q.$ (The function $\varphi$ is often called a potential function of $q$).
\end{enumerate}
\end{theorem}
This characterization helps us understand some critical aspects of cyclical monotonicity.
Let $\varphi_d^\ast (\cdot, x)$ be a potential function of $q_d^\ast (\cdot, x).$
Then, by the convexity of $\varphi_d^\ast (\cdot, x)$, it follows that
$$
    \frac{\partial q_d^{\ast i}(u, x)}{\partial u^i}
    =
    \frac{\partial^2 \varphi_d^\ast (u, x)}{\partial^2 u^i}
    \geq
    0
    ,
$$ 
where $q_d^\ast = (q_d^{\ast 1}, \dots, q_d^{\ast p})^\prime$ and $u = (u^1, \dots, u^p)^\prime.$
Moreover, the $i$th component of $U_d$ may also influence the $j$th component of $Y_d$ for $j \neq i.$
Specifically, the cross-partial derivative 
$$
    \frac{\partial q_d^{\ast j}(u, x)}{\partial u^i}
    =
    \frac{\partial^2 \varphi_d^\ast (u, x)}{\partial u^i \partial u^j}
$$
may be either positive or negative.
This feature enables the model to capture a flexible correlation structure across different dimensions of the outcome vector.

Another desirable feature of cyclical monotonicity is that probability distributions on $\bR^p$ are completely characterized by cyclically monotone maps.
More specifically, it is known that for any probability distribution $\nu$ on $\bR^p,$ there uniquely exists a cyclically monotone function $q : \cU \to \bR^p$ such that $q_\# \mu = \nu.$
Thanks to this property, we can perfectly capture the joint distribution of the $p$-dimensional potential outcomes with cyclically monotone maps, exactly as univariate quantile functions do.

Because a cyclically monotone function achieves this exact distributional matching while inherently preserving a structured ranking, we can interpret the multivariate model in parallel to the classical univariate case. 
Recall that in the one-dimensional framework, evaluating the structural function at a scalar rank $u \in [0,1]$ yields the potential outcome level for an individual at that specific percentile of the unobserved heterogeneity distribution. 
Based on cyclical monotonicity, we apply this classical quantile interpretation directly to our multivariate structural functions. 
Specifically, for a $p$-dimensional vector $u \in \cU,$ which represents a \textit{multivariate rank}, we can think of $q_d^\ast (u, x)$ as the ``$u$-quantile'' (with respect to $\mu$) of the distribution of $Y_d$ conditional on $X = x.$ 
It represents the vector-valued potential outcome level corresponding to an individual located at that joint position in the latent distribution. 
In this sense, the $u$-quantile is exactly the level of the potential outcome (cyclically) monotonically transformed from the multivariate latent rank $U_d = u,$ allowing our model to scale the standard quantile interpretation directly from an individual scalar rank to an individual multi-dimensional trait profile.

The present work is not the first one that generalizes the notion of quantiles in this direction.
The same idea has been developed by \cite{carlier2016vector}, \cite{chernozhukov2017monge}, \cite{hallin2021distribution}, and \cite{ghosal2022multivariate}.
Following the literature, we call a function $q$ the \textit{multivariate quantile function} of $\nu$ associated with $\mu$ if $q$ is cyclically monotone and $q_\# \mu = \nu.$
Further details are provided in Appendix \ref{ivmqr:app:preliminaries}.
}

Now, we summarize the set of assumptions imposed on the model (\ref{ivmqr:eq:structural-equation}).
\begin{assumption} \label{ivmqr:ass:ivqr}
\mbox{}
\begin{enumerate} [label=\textcolor{blue}{(A\arabic*)}, ref=A\arabic*]
    \item \label{ivmqr:ass:ivqr-A1} For each $d \in \cD,$ it holds that $U_d \mid X \sim \mu,$ and \memo{red}{for each $x \in \cX,$ the function $q_d^\ast (\cdot, x)$ is cyclically monotone.}
    \item \label{ivmqr:ass:ivqr-A2} For each $d \in \cD,$ $U_d$ is independent of $Z$ conditional on $X.$
    \item \label{ivmqr:ass:ivqr-A3} Conditional on $(X, Z, \nu),$ $(U_d)_{d \in \cD}$ are identically distributed.
\end{enumerate}    
\end{assumption}

\memo{red}{
Now, we refer to the structural model (\ref{ivmqr:eq:structural-equation}) with Assumption \ref{ivmqr:ass:ivqr} as an \textit{IV multivariate quantile regression model}.
}

Conditions (\ref{ivmqr:ass:ivqr-A2}) and (\ref{ivmqr:ass:ivqr-A3}) are identical to the assumptions in \cite{chernozhukov2005iv}.
Condition (\ref{ivmqr:ass:ivqr-A2}) states that the IV is independent of potential outcomes conditional on covariates.
\memo{red}{
Combined with (\ref{ivmqr:ass:ivqr-A1}), it implies $U_d \indep (X, Z).$
}

Condition (\ref{ivmqr:ass:ivqr-A3}) is called rank similarity by \cite{chernozhukov2005iv}.
The simplest form of rank similarity is rank invariance, which requires $U_d = U$ for all $d \in \cD.$
In the wealth accumulation example discussed in Section \ref{ivmqr:sec:introduction}, rank invariance assumes that an individual's preference does not change regardless of their participation in 401(k).
On the other hand, rank similarity assumes that the preference may change depending on the participation status, but the individual cannot predict the change before deciding whether to participate in 401(k) or not.
For more details of these assumptions, see \cite{chernozhukov2005iv} and \cite{chernozhukov2013quantile}.

An important implication of rank similarity is that the realized heterogeneity $U \coloneqq U_D$ is independent of the instrument $Z$ conditional on $X.$

\begin{proposition} \label{ivmqr:prop:IV-independence}
Under Assumption \ref{ivmqr:ass:ivqr}, $U \mid Z, X \sim \mu$ holds.
In particular, $U \indep Z \mid X.$
\end{proposition}

This result establishes the independence of the instrument $Z$ from the unobserved heterogeneity $U.$
When $Z$ is also relevant for $D,$ Proposition \ref{ivmqr:prop:IV-independence} supplies the exogeneity component of instrument validity.

\begin{remark}
The maps $\varphi_d^\ast (\cdot, x)$ and $q_d^\ast (\cdot, x)$ ($= \nabla \varphi_d^\ast (\cdot, x)$) are respectively known as an optimal transport potential and the optimal transport map from $\mu$ to the conditional distribution of $Y_d \mid X = x,$ denoted by $\cL (Y_d \mid X = x).$
In particular, the optimal transport map $q_d^\ast (\cdot, x)$ describes the cheapest way to ``transport'' mass from the reference distribution $\mu$ to the target distribution $\cL (Y_d \mid X = x)$.
We provide a more detailed mathematical overview of these concepts in Appendix \ref{ivmqr:app:preliminaries}.
We leverage this optimal transport framework extensively for identification and estimation in the later sections.    
\end{remark}

\begin{remark}
\memo{red}{
To contextualize the optimal transport-based framework, it is helpful to compare it with other multivariate quantile notions. 
Early extensions, such as geometric quantiles (\cite{chaudhuri1996geometric}, \cite{chakraborty2003multivariate}) and M-quantiles (\cite{koltchinskii1997m}), utilize $L_1$ optimization to represent quantiles as gradients of convex potentials. 
However, for dimensions more than one, these methods typically lack a representation like $Y_d = q_d (U_d, X)$ where $U_d$ follows a fixed reference distribution. 
Similarly, directional and projection-based quantiles (\cite{wei2008approach}, \cite{hallin2010multivariate}, \cite{kong2012quantile}) characterize regions as envelopes of hyperplanes rather than providing a global functional mapping from a reference space.
\cite{white2015var} take another approach by jointly modeling several scalar conditional quantiles to study dependence across multiple outcomes.
}

\memo{red}{
In contrast, the optimal transport-based approach (\cite{carlier2016vector}, \cite{chernozhukov2017monge}, \cite{hallin2021distribution}, \cite{ghosal2022multivariate}) defines the multivariate quantile as the unique gradient of a convex potential pushing a fixed reference distribution to the target. 
Its defining advantage is the strong representation $Y_d = q_d (U_d, X).$
This representation ensures a multidimensional form of monotonicity and allows for the joint modeling of outcomes.
This is essential for characterizing unobserved heterogeneity and capturing complex cross-dimensional dependencies in nonseparable structural models.
}
\end{remark}

Below, we provide two examples that fit in the present framework in the context of treatment effect estimation and discrete choice analysis.
\memo{red}{We acknowledge that our setup does not cover models where the structural function is not cyclically monotone (or not the gradient of a convex function).}

\begin{example}
We first consider an extension of the saving model proposed by \cite{chernozhukov2004effects}.
The paper investigates the effects of 401(k) participation on wealth accumulation.
As measures of wealth, they consider total wealth, net financial assets, and net non-401(k) financial assets.
Since \cite{chernozhukov2004effects} apply the IVQR model to these variables separately, their analysis does not capture the correlation among them.
For example, while total wealth is defined as net financial assets plus other assets (e.g., housing equity and the value of business, property, and motor vehicles), their model has no implication for how people distribute total wealth into these subcategories.
We shall see that our multivariate model can take both variables into consideration simultaneously and allows us to discuss their joint distribution.

Let $U_d = (U_d^F, U_d^O)^\prime \sim \mu$ represent the preference for net financial assets and other assets under participation status $d \in \{0, 1\}.$
Let $y_d = (y_d^F, y_d^O)^\prime$ be net financial assets and other assets, respectively.
An individual with type $U_d$ whose portfolio is $y_d$ receives a quasi-linear utility $U_d^\prime y_d + \psi_d (y_d),$ where $\psi_d$ is a deterministic concave function, suppressing covariates $X$ for simplicity.
Assuming an interior solution, the individual chooses the optimal portfolio $Y_d = (- \nabla \psi_d)^{-1} (U_d),$ which follows by the first-order condition.
If the convex conjugate of $- \psi_d$ is denoted by $\varphi_d,$\footnote{For a convex function $f : \cU \to \bR,$ its convex conjugate $g$ is defined as $g (v) \coloneqq \sup_{u \in \cU} (u^\prime v - f (u)).$ It is easy to see that $g$ is convex. Assuming sufficient regularity, $(\nabla f)^{-1} = \nabla g$ holds. See, for example, \cite{rockafellar1970convex} for more details.} the potential outcome is written as $Y_d = \nabla \varphi_d (U_d),$ where $\nabla \varphi_d$ is cyclically monotone.
\memo{red}{We can understand this function as a mapping from the latent variable $U_d,$ which represents unobserved preferences, to the level of a portfolio vector $Y_d$ in a ``monotonic'' way.}

The individual determines whether to participate in 401(k) based on the eligibility $Z$ of 401(k) and unobservable heterogeneity $\nu.$
This is represented by  $D = \delta (Z, \nu).$
The eligibility variable $Z$ is taken to be a valid IV.
The joint distribution of $(Y \coloneqq Y_D, D, Z)$ fits into the multivariate IVQR model under $Z \indep U_d$ and rank invariance/similarity.

With $q_d$ identified, we can obtain the joint distribution of net financial assets and other assets.
This allows us to discuss the effect of 401(k) participation on the portfolio selection, which is impossible in the standard IVQR framework.
\end{example}

\begin{example}
We now show how a simple discrete choice problem fits naturally into the IV multivariate quantile framework.
This example is inspired by \cite{shi2018estimating} and \cite{fosgerau2020discrete}.
Consider a retailer offering $p$ differentiated goods in two sets of stores, indexed by $d \in \{0, 1\}.$
\memo{red}{The products' qualities are captured by a latent quality vector $U = (U^1, \dots, U^p)^\prime \sim \mu,$ where $\mu$ is a prespecified reference distribution. 
External quality measures, such as in-store audits, may help motivate the specification of $\mu$, although identifying $\mu$ from such measurements would require an additional measurement model, which we do not impose here.
}
Each consumer in district $d$ has an unobservable preference vector $\varepsilon = (\varepsilon^0,\varepsilon^1,\ldots,\varepsilon^p)^\prime \sim \eta_d.$
The preference distributions $\eta_0, \eta_1$ are unknown and may differ, for example, because district $1$ is exposed to a marketing campaign.
The utility from the outside option is $\varepsilon^0$, while the utility from inside good $i$ is $u^i + \varepsilon^i.$

The analyst never observes $U$ or $\varepsilon$. 
Instead, one observes the aggregate market share vector $s_d (U) \coloneqq (s_d^1 (U), \dots, s_d^p (U))$ if the store is in district $d,$ where
$$
  s_d^i(u)
  \coloneqq
  \int
    \indic
    \left\{
    u^i+\varepsilon^i
    =
    \max
    \left(
        \varepsilon^0,
        u^1+\varepsilon^1,
        \ldots,
        u^p+\varepsilon^p
    \right)
\right\}
  d \eta_d(\varepsilon)
  .
$$
Economically, $s_d (u)$ represents how a marginal improvement in the quality $u$ redistributes sales across all $p$ goods under the taste regime $\eta_d.$

The Williams-Daly-Zachary theorem implies that the surplus function defined as
$$
    W_d (u)
    \coloneqq
    \int
        \max
        \left(
            \varepsilon^0,
            u^1+\varepsilon^1,
            \ldots,
            u^p+\varepsilon^p
        \right)
    d \eta_d (\varepsilon)
$$
is related to the market share vector via the equation $s_d (u) = \nabla W_d (u)$ (see, e.g., \cite{mcfadden1981econometric}).
It is not hard to see that $W_d$ is a convex function and differentiable under mild regularity on $\eta_d.$ 
\memo{red}{As in the previous example, $\nabla W_d$ is cyclically monotone, and it bridges the latent quality vector $U$ with the product share $s_d (U)$ monotonically.}

Let $D$ denote the district, or equivalently the treatment status, ($0$ or $1$) and suppose there exists a valid instrument $Z$ (e.g. randomized ad‐exposure eligibility). Then we observe i.i.d. draws of $(Y, D, Z)$ with $Y = s_D(U) = \nabla W_D (U),$ $U \sim \mu,$ and $Z \indep U.$
This setup satisfies rank invariance and exogeneity conditions of the IV multivariate quantile model.

Applying a separate IVQR model of \cite{chernozhukov2005iv} to each product discards the rich cross‐product relationships.
For instance, one-dimensional approaches cannot impose the natural constraint $\sum_{i=1}^p s_d^i(u) \leq 1,$ nor capture correlations in $U$ arising from similar products facing related quality shocks.  
The multivariate model preserves these economically meaningful cross‐product patterns while identifying the two demand-quality mappings $s_0$ and $s_1.$
\end{example}

One limitation of the IV multivariate quantile regression model is that Condition (\ref{ivmqr:ass:ivqr-A3}) depends on the choice of the reference measure $\mu.$
In particular, the fact that Condition~(\ref{ivmqr:ass:ivqr-A3}) holds for a given $\mu$ does not guarantee that it will hold for another choice $\tilde{\mu}.$
This dependence implies that the distribution $\mu$ of unobserved heterogeneity must be specified a priori.\footnote{\memo{red}{If $U \sim \mu$ satisfies Condition (\ref{ivmqr:ass:ivqr-A3}), then so does the distribution of $a U + b$ for any $a > 0$ and $b \in \bR^p.$ Hence, assuming that $\cU \subset [0, 1]^p$ is without loss.}}

\memo{red}{
To see why a change of reference measure is not generally innocuous, consider a common transformation $T$ with $T_{\#}\mu=\tilde\mu$ and define $\tilde U_d \coloneqq T(U_d)$ and $\tilde q_d \coloneqq q_d\circ T^{-1}$.
This transformation preserves the observable outcomes and rank similarity, but it yields another admissible representation only when $\widetilde q_d$ is cyclically monotone for every $d$.
In one dimension, an increasing transformation $T$ preserves monotonicity, whereas in multiple dimensions cyclical monotonicity is not generally preserved.
Thus, unlike in the scalar quantile model, arbitrary changes of the reference measure need not lead to observationally equivalent models within our model class.
}

Although restrictive, specifying the distribution of unobserved heterogeneity is common in related literature.
For example, \cite{chernozhukov2021identification} impose a similar condition for identifying hedonic equilibrium models, and discrete choice models also often assume a fixed distribution for unobserved heterogeneity (\cite{train2009discrete}).

Following the literature, we say that a vector of structural functions is observationally equivalent if it is consistent with the observed random variables.

\begin{definition} \label{ivmqr:def:observational-equivalence}
A vector of functions $\tilde q = (\tilde q_d)_{d \in \cD}$ is \textit{observationally equivalent} to $q^\ast$ if there is a set of random variables $((\tilde Y_d)_{d \in \cD}, \tilde D, \tilde X, \tilde Z, (\tilde U_d)_{d \in \cD})$ that satisfies the structural equation (\ref{ivmqr:eq:structural-equation}) and Assumption \ref{ivmqr:ass:ivqr} with $\tilde q,$ and the joint distribution of $(\tilde Y \coloneqq \tilde Y_{\tilde D}, \tilde D, \tilde X, \tilde Z)$ is equal to that of $(Y, D, X, Z).$
\end{definition}

By condition (\ref{ivmqr:ass:ivqr-A1}), the observed outcome is written as the structural form $Y = q_D^\ast (U, X)$ where $U \coloneqq U_D.$ 
The following representation is the main testable implication of the model.

\begin{theorem} \label{ivmqr:thm:representation}
Suppose that the support of $Y_d$ is convex for each $d.$
\memo{red}{Let $q = (q_d)_{d \in \cD}$ be such that $q_d (\cdot, x)$ is strictly cyclically monotone.}
Then, $q$ is observationally equivalent to $q^\ast$ if and only if with probability one, 
\begin{align} \label{ivmqr:eq:testable-representation}
    \bP (
        Y \in q_D (B, X)
        \mid 
        X, Z
    ) 
    = 
    \mu (B)
\end{align}
holds for every measurable $B \subset \cU,$ where $q_D (B, X) = \{q_D (u, X) \mid u \in B\}.$
\end{theorem}

\memo{red}{
Here, $\mu(B)$ has a direct quantile interpretation.
Since $\mu$ is the distribution of the multivariate rank, $\mu(B)$ is the probability mass associated with the rank region $B$.
Thus, equation (\ref{ivmqr:eq:testable-representation}) states that the image $q_D^\ast(B,X)$ contains exactly the outcome realizations corresponding to this mass of latent ranks, conditional on the instrument and covariates.
}

As the LHS of (\ref{ivmqr:eq:testable-representation}) is determined by the joint distribution of the observable variables, the equation gives a conditional moment restriction.
In particular, any candidate structural function $(q_d)_{d \in \cD}$ must satisfy (\ref{ivmqr:eq:testable-representation}).
For identification, we explore conditions under which $(q_d^\ast)_{d \in \cD}$ is the only function that satisfies the equation in the subsequent section.

For notational simplicity, let $q_d^\ast(u) \coloneqq q_d^\ast(u, x)$, suppressing the dependence on covariates $x$.
All subsequent analysis should be understood as conditional on $X = x$.
This simplification does not affect the identification results, since the conditional distribution of the observables given $X$ is identified.

When outcomes are univariate and the reference distribution is the uniform distribution on the unit interval, as in \cite{chernozhukov2005iv}, $q_d^\ast (\tau)$ is the $\tau$-quantile of $Y_d$ and satisfies
$$
    \bP (
        Y \leq q_D^\ast (\tau)
        \mid 
        Z = z
    ) 
    = 
    \tau
$$
for all $z \in \cZ,$ which corresponds to setting $\mu = U [0, 1]$ and $B = [0, \tau]$ in (\ref{ivmqr:eq:testable-representation}).
If the treatment variable is supported on a finite set $\cD = \{0, \dots, m - 1\}$, this is a simultaneous equation with $m$ unknown variables $q_0^\ast (\tau), \dots, q_{m - 1}^\ast (\tau).$
Hence, if the instrument $Z$ takes at least $m$ values and the equation system is non-degenerate, it is expected that the solution is unique, at least locally.
Indeed, \cite{chernozhukov2005iv} show that this observation is correct.
Checking the condition for each quantile level $\tau$ establishes the nonparametric identification of the structural function $q_d.$

This approach works in one dimension because, for each $\tau$, the values $(q_d^\ast(\tau))_{d \in \cD}$ are determined by the conditional moment restriction (\ref{ivmqr:eq:testable-representation}) when we take the set $B = [0, \tau]$.
In higher dimensions, however, this argument no longer applies.
To see why, consider the case $p = 2$ with $\mu = U[0,1]^2$ and fix $(\tau^1, \tau^2) \in [0,1]^2$.
For any set $B \subset [0,1]^2$, equation (\ref{ivmqr:eq:testable-representation}) does not isolate $(q_d^\ast(\tau^1, \tau^2))_{d \in \cD}$ alone, because the image $q_d^\ast(B)$ can be fully nonlinear.
In particular, the image depends not only on the value of $q_d^\ast$ at $(\tau^1, \tau^2)$ but also on its behavior along the boundary $\partial B$.
Since the set $q_d^\ast(\partial B)$ is an infinite-dimensional object, evaluating (\ref{ivmqr:eq:testable-representation}) for a fixed $B$ does not yield identification when the instrument has finite support.
To address this issue, we instead interpret (\ref{ivmqr:eq:testable-representation}) as a measure-valued equation, allowing us to exploit the relationships between $q_d^\ast(B)$ and $q_d^\ast(B')$ for different sets $B$ and $B'$.
Such relationships are unnecessary in the univariate case but are crucial for identifying multivariate structural functions.

\subsection{Relevant Models}

\textbf{Univariate IVQR}.
As discussed above, even when multiple outcome variables are present, one might consider applying the one-dimensional quantile model of \cite{chernozhukov2005iv} to each component separately.
However, this approach typically leads to the violation of rank similarity assumption of \cite{chernozhukov2005iv}.
To see this, we revisit the wealth accumulation example.
For $p = 2,$ consider a binary treatment environment and assume that $(Y_0, Y_1, D, Z, U_0, U_1)$ satisfies Assumption \ref{ivmqr:ass:ivqr} with rank invariance $U = U_0 = U_1$ and $D = \delta (Z, U, \nu),$ where $\nu$ is a random variable independent of all the other variables.
Recall that $Y_d = (Y_d^F, Y_d^O)^\prime$ is the vector consisting of net financial assets and other assets, $D$ is the participation status of 401(k), $Z$ is its eligibility, and $U = (U^F, U^O)^\prime$ is the preference for the corresponding assets.
The potential outcome for treatment $d$ is represented component-wise as
$$
    Y_d^F
    =
    q_d^{\ast F} (U)
    \quad
    \text{and}
    \quad
    Y_d^O
    =
    q_d^{\ast O} (U)
    .
$$
Since $U$ and $Z$ are assumed independent, it is expected that the treatment effect of $D$ is identified if $Z$ is sufficiently informative.

Consider an empirical researcher who is interested only in net financial assets $Y_d^F.$
In this case, it is common to apply the standard quantile model to the data $(Y^F, D, Z),$ as \cite{chernozhukov2004effects} do.
To see that this practice leads to the violation of rank similarity assumption, notice that the structural equation under the one-dimensional model is $Y_d^F = \tilde q_d (\tilde U_d),$ where $\tilde U_d \sim U [0, 1]$ and $\tilde q_d$ is the quantile function of $Y_d^F.$
Then, $\tilde U_d$ is measurable in $U$ since $\tilde U_d = \tilde q_d^{-1} (q_d^{\ast F} (U)).$
Rank similarity requires $\tilde U_0 \overset{d}{=} \tilde U_1 \mid Z, U, \nu,$ but this is satisfied only when $\tilde q_0^{-1} \circ q_0^{\ast F} = \tilde q_1^{-1} \circ q_1^{\ast F},$ which does not hold in general.\footnote{For a counterexample, consider $q_0^\ast (u) = u$ and $q_1^\ast (u) = ((u^1 + u^2) / 2, (u^1 + u^2) / 2)^\prime.$}
In other words, the instrument $Z$ is invalid in the sense that it is correlated with the unobserved heterogeneity $\tilde U \coloneqq \tilde U_D$ via the treatment variable $D.$

The failure of rank similarity arises because $q_d^{\ast F}$ depends not only on $U^F$ but also on $U^O,$ whereas \cite{chernozhukov2005iv} do not allow such dependence, which is an assumption that is often unrealistic.
Specifically, applying the framework of \cite{chernozhukov2005iv} to the wealth accumulation example would require assuming that individuals determine their holdings of financial assets based solely on their preference for financial assets.
In reality, individuals also consider their preference for other assets, since it affects how much they can allocate to financial assets.
Our multivariate model accommodates such interdependence by allowing the potential outcomes to depend jointly on multiple unobservables $U = (U^F, U^O)^\prime,$ thereby resolving this limitation.

In summary, the IV multivariate quantile model enables researchers to consider multiple outcome variables at the same time, since it captures the correlation across different outcomes.
Moreover, even if only some of the outcome variables are of interest, our model can alleviate the endogeneity by considering other outcomes together.

\bigskip

\noindent
\memo{red}{\textbf{Other Triangular Models}}.
\memo{red}{
Many papers have studied similar triangular models for both uni- and multivariate outcomes under different sets of assumptions. 
\cite{imbens2009identification} assume a control function approach where the first stage is strictly monotone in a scalar unobservable and the instrument is jointly independent of both structural and selection disturbances. 
\cite{torgovitsky2015identification} and \cite{d2015identification} consider nonseparable models that require both the outcome and first-stage equations to be strictly increasing in scalar unobservables, enabling identification even with instruments having small support. 
\cite{hoderlein2017corrigendum} clarify that identifying structural objects in such systems generally requires first-stage heterogeneity to be one-dimensional and monotonic.
\cite{gunsilius2023condition} extends this literature to models with continuously distributed multivariate endogenous variables and binary instruments by imposing cyclical monotonicity on the first-stage map.
}

\memo{red}{
The present paper is different from this stream of papers. 
Firstly, it relaxes the requirement of strong exogeneity, which typically assumes the instrument is jointly independent of the disturbances in both the outcome and selection equations.
Instead, consistent with the framework of \cite{chernozhukov2005iv}, it only requires that potential outcome ranks are independent of the instrument, allowing for complex selection mechanisms where instruments may relate to components of the selection process. 
Secondly, it allows for a very general first-stage equation where treatment selection can depend on arbitrary randomness rather than being restricted to a scalar unobservable or requiring monotonicity conditions.
}

\memo{red}{
At the cost of these relaxations, the identification of the structural functions requires a completeness-type condition. 
This condition serves as a multivariate generalization of the standard full-rank condition used in one-dimensional instrumental variable quantile models.
It requires variation in the conditional distribution of $(Y, D)$ across $Z.$
See Assumption \ref{ivmqr:ass:div-free} for more details.
}

\section{Identification} \label{ivmqr:sec:identification}

In this section, we consider the identification problem of the structural functions.
We assume that the treatment and IV are supported on the same finite set, i.e., $\cD = \cZ = \{0, 1, \dots, m - 1\}.$

The representation (\ref{ivmqr:eq:testable-representation}) of Theorem \ref{ivmqr:thm:representation} and the change-of-variables formula imply that the true structural function $q^\ast = (q_0^\ast, \dots, q_{m - 1}^\ast)^\prime$ solves the measure-valued equation
\begin{equation} \label{ivmqr:eq:system-monge-ampere}
    \mu (d u)
    =
    \sum_{d \in \cD}
    f_{d, z} (q_d (u))
    \det (D q_d (u))
    d u
    \ \ \text{for} \
    z 
    \in 
    \cZ
    ,
\end{equation}
where 
\begin{align*}
    f_{d, z} (y) 
    &\coloneqq 
    \frac{\partial^p}{\partial y^1 \dots \partial y^p} 
    \bP (Y \leq y, D = d \mid Z = z)
    \\
    &=
    \left(
        \frac{\partial^p}{\partial y^1 \dots \partial y^p} 
        \bP (Y \leq y \mid D = d, Z = z)
    \right)
    \bP (D = d \mid Z = z)
    .
\end{align*}

For a scalar function $a : \cU \to \bR,$ a vector-valued function $b : \cU \to \bR^p,$ and a matrix-valued function $M : \cU \to \bR^{p \times p},$ we define three types of supremum norms as follows:
\begin{align*}
    \norm{a}_\infty
    &\coloneqq
    \sup_{u \in \cU}
    |a (u)|
    ,
    \\
    \norm{b}_\infty
    &\coloneqq
    \max_{1 \leq i \leq p}
    \norm{b^i}_\infty
    ,
    \\
    \norm{M}_\infty
    &\coloneqq
    \sup_{u \in \cU}
    \sup_{
        \substack{
            v \in \bR^p \\
            \norm{v} = 1
        }
    }
    \norm{M (u) v}
    .
\end{align*}

For an open set $V$ in a Euclidean space, let $C^k (V; \bR^s)$ be the space of $k$-times continuously differentiable $\bR^s$-valued functions on $V.$
For a compact set $K$ in a Euclidean space, define
$$
    C^k (K; \bR^s)
    \coloneqq
    \left\{
        f : K \to \bR^s
        \mid
        \exists \text{ open } V \supset K
        \text{ and }
        g \in C^k (V; \bR^s)
        \text{ s.t. }
        g \big |_{K} = f
    \right\}
    .
$$
Finally, let $C^\infty (K; \bR^s) \coloneqq \cap_{k = 0}^\infty C^k (K; \bR^s).$

Our identification result below shows that $q = q^\ast$ is the locally unique solution to (\ref{ivmqr:eq:system-monge-ampere}) in a certain function class.
To state the result formally, consider the normed space $\cQ \coloneqq \left(C^2 (\cU; \bR^p)\right)^m$ equipped with $\norm{q}_\cQ \coloneqq \max_{d \in \cD} \norm{q_d}_\infty.$
Also, for fixed constants $\overline \lambda > \underline \lambda > 0,$ define a subset $\tilde \cQ$ of $\cQ$ as
\begin{align*}
    \tilde \cQ
    \coloneqq
    \left\{
        q \in \cQ
        \ \ \Big | \
        \begin{array}{l}
            \text{(i)} \ \memo{red}{q_d \text{ is cyclically monotone}} \\
            \text{(ii)} \ \underline \lambda \leq \lambda_{\text{min}} (D q_d) \leq \lambda_{\text{max}} (D q_d) \leq \overline \lambda
        \end{array}
    \right\}
    ,
\end{align*}
where $\lambda_{\text{max}} (A)$ is the largest eigenvalue of $A,$ and $\lambda_{\text{min}} (A)$ is defined similarly.
For a constant $K > 0,$ the space in which the parameter is to be identified is
\begin{align*}
    \cQ_0
    \coloneqq
    \left\{
        q^\ast + \alpha h \in \tilde \cQ
        \mid
        \alpha \geq 0
        ,
        \norm{h}_\cQ = 1
        ,
        \max_{d \in \cD} \norm{D h_d}_\infty \leq K
    \right\}
    .
\end{align*}

Since we are interested in identification, we assume the correct specification.
\begin{assumption} \label{ivmqr:ass:correct-specification}
$q^\ast \in \cQ_0.$
\end{assumption}

Under Assumption \ref{ivmqr:ass:ivqr} (\ref{ivmqr:ass:ivqr-A1}), Assumption \ref{ivmqr:ass:correct-specification} restricts $q_d^\ast$ to be twice continuously differentiable and have moderate derivatives.

The conditions in the definition of $\cQ_0$ restrict the class of admissible deviations from the true structural function.
Such restrictions are common in the literature on nonparametric IV identification.
To see their role, note that identifying $q$ requires inverting the linearization of the system of equations (\ref{ivmqr:eq:system-monge-ampere}).
Although this linearized system is invertible under the standard full-rank condition, its inverse is generally discontinuous because of the infinite-dimensional nature of the problem, rendering the identification ill-posed.
This difficulty is resolved by restricting the identification domain to $\cQ_0,$ thereby ruling out pathological deviations from the truth.
Similar regularity conditions are imposed in related work, including \cite{chen2014local} and \cite{centorrino2024iterative}.
In particular, Section 2.3 of \cite{chen2014local} presents an example where identification fails once an analogous restriction is removed.

We further impose two regularity conditions.
Let $\cY_d \subset \bR^p$ be the support of the distribution of $Y_d.$
\begin{assumption} \label{ivmqr:ass:regular-boundary}
For each $d \in \cD,$ the support $\cY_d$ is a convex compact set.
\end{assumption}

\begin{assumption} \label{ivmqr:ass:smooth-densities}
For each $(d, z) \in \cD \times \cZ,$ $f_{d, z} \in C^1 (\cY_d; \bR).$
\end{assumption}

Assumption \ref{ivmqr:ass:regular-boundary} requires that the support of the potential outcomes is regular.
Assuming a convex support is plausible in many applications, and it is typically necessary in the optimal transport theory, which we rely upon for identification and estimation; see, for example, \cite{villani2003topics} and \cite{figalli2017monge}.
Combined with Assumption \ref{ivmqr:ass:regular-boundary}, Assumption \ref{ivmqr:ass:smooth-densities} implies that $f_{d, z}$ is bounded.

We consider the identification of structural functions in the following sense.
\begin{definition} \label{ivmqr:def:identification}
Suppose $q^\ast \in \cQ_0.$
We say the structural function $q^\ast$ is \textit{identified} in $\cQ_0$ if $q \in \cQ_0$ being observationally equivalent to $q^\ast$ implies $q = q^\ast.$
Also, we say the structural function $q^\ast$ is \textit{locally identified} in $\cQ_0$ if there exists $\varepsilon > 0$ such that $q^\ast$ is identified in $\{q \in \cQ_0 \mid \norm{q - q^\ast}_\cQ < \varepsilon\}.$
\end{definition}

For an invertible matrix $A,$ let $\cof (A) \coloneqq \det (A) (A^{-1})^\prime$ be its cofactor matrix.
We introduce a key assumption.
We discuss this assumption in detail below.

\begin{assumption} \label{ivmqr:ass:div-free}
If $\kappa = (\kappa_d)_{d \in \cD} \in (C^1 (\cU))^m$ satisfies the following condition, then $\nabla \kappa_d \equiv 0$ for all $d \in \cD:$
\begin{equation} \label{ivmqr:eq:div-free}
    \text{For any $\psi \in C^\infty (\cU)$ and $z \in \cZ,$ }
    \int_\cU 
        (\nabla \psi (u))^\prime
        \left(
            \sum_{d \in \cD}
            f_{d, z} (q_d^\ast (u))
            \cof (D q_d^\ast (u))
            \nabla \kappa_d (u)
        \right)
    d u
    =
    0
    .
\end{equation}
\end{assumption}

The following is the main theorem of this section.

\begin{theorem} \label{ivmqr:thm:local-identification-IVQR-binary}
Suppose that Assumptions \ref{ivmqr:ass:ivqr}, \ref{ivmqr:ass:correct-specification}, \ref{ivmqr:ass:regular-boundary}, \ref{ivmqr:ass:smooth-densities}, \ref{ivmqr:ass:div-free} hold.
Then, $q^\ast$ is locally identified in $\cQ_0.$ 
\end{theorem}

Assumption \ref{ivmqr:ass:div-free} is the key condition for identification.
This can be viewed as a version of the completeness condition utilized in \cite{newey2003instrumental} and \cite{chernozhukov2007instrumental}.
The integral in condition (\ref{ivmqr:eq:div-free}) represents the linearization of the identification equation (\ref{ivmqr:eq:system-monge-ampere}) around $q = q^\ast,$ where the function $\nabla \kappa_d$ is interpreted as a possible deviation from the true structural function $q^\ast.$
This assumption requires that the truth $(\nabla \kappa_d \equiv 0)$ is the unique function that makes the identification equation (\ref{ivmqr:eq:system-monge-ampere}) first-order insensitive, which establishes identification.

The restrictions that Assumption \ref{ivmqr:ass:div-free} imposes on the primitives differ depending on whether the outcome variable is one-dimensional ($p = 1$) or multi-dimensional ($p \geq 2$).

We first consider the case when the outcome is one-dimensional with $\mu = U [0, 1].$
In this case, the model satisfying Assumption \ref{ivmqr:ass:ivqr} coincides with the standard IVQR model of \cite{chernozhukov2005iv}.
Notice that for any $1 \times 1$ matrix (i.e., scalar) $A,$ its cofactor matrix is trivial, $\cof (A) = 1,$ and that the derivative $\nabla \psi = \psi^\prime$ can take any value in $C^\infty [0, 1].$
Consequently, condition (\ref{ivmqr:eq:div-free}) implies that for each $z \in \cZ,$
\begin{equation*}
    \int_0^1
        \chi (u)
        \left(
            \sum_{d \in \cD}
            f_{d, z} (q_d^\ast (u))
            \kappa_d^\prime (u)
        \right)
    d u
    =
    0
\end{equation*}
holds for any $\chi \in C^\infty [0, 1].$
By the fundamental lemma of the calculus of variations, it follows that the function inside the parentheses is identically zero on $[0, 1].$
Therefore,
$$
    \underbrace{
        \begin{pmatrix}
            f_{0, 0} (q_0^\ast (u)) & \dots & f_{m - 1, 0} (q_{m - 1}^\ast (u)) \\
            \vdots & \ddots & \vdots \\
            f_{0, m - 1} (q_0^\ast (u)) & \dots & f_{m - 1, m - 1} (q_{m - 1}^\ast (u))\\
        \end{pmatrix}
    }_{\eqqcolon \partial \Pi_m (q^\ast (u))}
    \begin{pmatrix}
        \kappa_0^\prime (u) \\
        \vdots \\
        \kappa_{m - 1}^\prime (u)
    \end{pmatrix}
    =
    \begin{pmatrix}
        0 \\
        \vdots \\
        0
    \end{pmatrix}
$$
holds for all $u \in [0, 1].$
If the matrix $\partial \Pi_m (q^\ast (u))$ is full-rank for all $u,$ then the $\kappa$ satisfies the above equation if and only if $\kappa_0^\prime \equiv \dots \equiv \kappa_{m - 1}^\prime \equiv 0.$
Therefore, Assumption \ref{ivmqr:ass:div-free} reduces to the matrix $\partial \Pi_m$ being full-rank.

This full-rank condition is obtained by \cite{chernozhukov2005iv}, \cite{chernozhukov2013quantile}, \cite{chernozhukov2017instrumental}.
In particular, when the treatment and IV are both binary (i.e., $m = 2$), the condition is written as 
$$
    \frac{f_{1, 1} (q_1^\ast (u))}{f_{0, 1} (q_0^\ast (u))}
    >
    \frac{f_{1, 0} (q_1^\ast (u))}{f_{0, 0} (q_0^\ast (u))}
    ,
$$
or the same inequality but with $>$ replaced with $<.$
This inequality requires the instrument $Z$ to have a monotone impact on the likelihood ratio of $(Y, D).$

When the outcome is multi-dimensional, Assumption \ref{ivmqr:ass:div-free} becomes more demanding; in particular, the standard full-rank condition is no longer sufficient to guarantee that the assumption holds in general.
This complication arises because, for $p \geq 2,$ the space of test functions $\nabla \psi$ is significantly smaller than the space of continuous functions on $\cU.$
Consequently, the fundamental lemma of the calculus of variations cannot be invoked to conclude that 
$$
    \sum_{d \in \cD}
    f_{d, z} (q_d^\ast (u))
    \cof (D q_d^\ast (u))
    \nabla \kappa_d (u)
    =
    0
    ,
$$
unlike in the one-dimensional case.
This means that non-zero deviations from the true structural function $q^\ast$ may not affect the identification equation (\ref{ivmqr:eq:system-monge-ampere}) in the first order.

To rule out such deviations, stronger conditions are required.
The following proposition provides a sufficient condition.

\begin{proposition} \label{ivmqr:prop:sufficient-condition-pd}
The condition that the matrix
$$
    \footnotesize
    F^\ast (u)
    \coloneqq
    \begin{pmatrix}
        f_{0, 0} (q_0^\ast (u)) \cof (D q_0^\ast (u)) & \dots & f_{m - 1, 0} (q_{m - 1}^\ast (u)) \cof (D q_{m - 1}^\ast (u)) \\
        \vdots & \ddots & \vdots \\
        f_{0, m - 1} (q_0^\ast (u)) \cof (D q_0^\ast (u)) & \dots & f_{m - 1, m - 1} (q_{m - 1}^\ast (u)) \cof (D q_{m - 1}^\ast (u))
    \end{pmatrix}
$$
is positive definite for almost all $u \in \cU$ is sufficient for Assumption \ref{ivmqr:ass:div-free}.
\end{proposition}

The condition in Proposition \ref{ivmqr:prop:sufficient-condition-pd} is more restrictive than the full-rank condition, even in the simplest case where $q_d^\ast (u) = u,$ $p = 2$ and the treatment and IV are binary ($m = 2$).
In this environment, it follows that
$$
    F^\ast (u)
    =
    \begin{pmatrix}
        f_{0, 0} (u) I_2 & f_{1, 0} (u) I_2 \\
        f_{0, 1} (u) I_2 & f_{1, 1} (u) I_2
    \end{pmatrix}
    =
    \partial \Pi_2 (u, u)
    \otimes
    I_2
    ,
$$
where $I_2$ is the $2 \times 2$ identity matrix.
This matrix is positive definite if and only if 
$$
    4 f_{0, 0} (u) f_{1, 1} (u)
    >
    \left(
        f_{0, 1} (u)
        +
        f_{1, 0} (u)
    \right)^2
$$
holds.
It is clear that this condition is stronger than the full rank of $\partial \Pi_2 (u, u)$, which is required when $p = 1.$

Since the condition in Proposition \ref{ivmqr:prop:sufficient-condition-pd} involves derivatives $D q_d^\ast,$ it is hard to test in practice.
Below, we provide a condition that can be tested from observed data, focusing on settings with binary treatments and IVs ($m = 2$).
Suppose that the density of $Y,$ conditional on $D = d$ and $Z = z,$ is bounded from both above and below by some constants $M > m > 0,$ respectively, then the condition
\begin{equation} \label{ivmqr:eq:sufficient-condition-correlation}
    \frac{\bP (D = 0 \mid Z = 0) \bP (D = 1 \mid Z = 1)}{(\bP (D = 0 \mid Z = 1) + \bP (D = 1 \mid Z = 0))^2}
    >
    \frac{M^2}{4 m^2}
    \left(\frac{\overline \lambda}{\underline \lambda}\right)^{p + 1}
\end{equation}
is sufficient for the matrix $F^\ast (u)$ to be positive definite.

\memo{red}{
Condition (\ref{ivmqr:eq:sufficient-condition-correlation}) is useful because its first-stage probabilities and conditional outcome densities are identified from the observed distribution.
The constants $\underline{\lambda}$ and $\bar{\lambda}$ characterize the maintained class of admissible structural functions.
Thus, given these maintained bounds, condition (\ref{ivmqr:eq:sufficient-condition-correlation}) provides an observable sufficient condition for Assumption \ref{ivmqr:ass:div-free}.
}

\memo{red}{
For a simple illustration, suppose that $P(D=0\mid Z=0)=P(D=1\mid Z=1)=\pi.$
Then, condition (\ref{ivmqr:eq:sufficient-condition-correlation}) reduces to
$$
\frac{\pi}{1-\pi}
>
\frac{M}{m}
\left(\frac{\bar{\lambda}}{\underline{\lambda}}\right)^{(p+1)/2}.
$$
Thus, the condition requires stronger first-stage separation when the conditional outcome densities are more heterogeneous.
In the benchmark case $M=m$ and $\bar{\lambda}=\underline{\lambda}$, it simply requires $\pi>1/2$.
}

\memo{red}{
When the treatment and IV are binary, Assumption \ref{ivmqr:ass:div-free} holds if there is an IV state that predicts the treatment perfectly.
\begin{proposition} \label{ivmqr:prop:sufficient-condition-zero}
Let $m = 2.$
For each $d,$ suppose that $D q_d^\ast$ exists with $\underline \lambda \leq \lambda_\text{min} (D q_d^\ast) \leq \lambda_\text{max} (D q_d^\ast) \leq \bar \lambda$ for some $0 < \underline \lambda < \bar \lambda,$ and that $\supp (\cL (Y \mid D = d, Z = d)) = \cY_d.$
If $\bP (D = 1 \mid Z = 0) = 0,$ then Assumption \ref{ivmqr:ass:div-free} holds.
\end{proposition}
}

\memo{red}{
This case arises when the IV represents the availability of the treatment; see, for example, \cite{abadie2002instrumental}.
In this design, the researcher does not need to verify the completeness condition directly.
The structure of the treatment-assignment mechanism itself guarantees it.
Note also that this condition does not imply the positive definiteness condition.
}

Remarkably, the condition that $\bP (D = 1 \mid Z = 0) = 0$ guarantees global, not just local, identification of the structural function.
The proof relies on the arguments developed in Section \ref{ivmqr:sec:identification-unobserved-distributions}.

\subsection{Outline of the proof of Theorem \ref{ivmqr:thm:local-identification-IVQR-binary}}

The proof of Theorem \ref{ivmqr:thm:local-identification-IVQR-binary} proceeds in two steps. 
First, Lemma \ref{ivmqr:lem:support-identification} shows that every observationally equivalent candidate in $\mathcal Q_0$ has a range containing $\mathcal Y_d$ for each $d.$
We may therefore restrict attention to a smaller class of structural functions.

Second, we show that $q^\ast$ is the locally unique solution to the system (\ref{ivmqr:eq:system-monge-ampere}) of equations within the class of structural functions.
The argument follows the logic of the implicit function theorem, which is commonly used in the identification analysis of nonlinear models with finite-dimensional parameters.
Specifically, we linearize the nonlinear system (\ref{ivmqr:eq:system-monge-ampere}) around the true structural function $q^\ast.$
The key insight is that the completeness condition (Assumption \ref{ivmqr:ass:div-free}) ensures that the ``slope'' of the linearized system is bounded away from zero, thereby guaranteeing the local uniqueness of the solution.

The full proof appears in Appendix \ref{ivmqr:sec:proof-identification-IVQR}.
The first step is established in Lemma \ref{ivmqr:lem:support-identification}.
The second step follows from a more general identification result presented in Appendix \ref{ivmqr:app:local-identification-system}.

\memo{red}{
Note that the fixed-point argument in \cite{torgovitsky2015identification} and \cite{gunsilius2023condition} cannot be applied directly to our setting.
That argument iterates a transformation defined over the support of a continuously distributed endogenous variable, whereas the treatment in our model is discrete.
}

\subsection{\memo{red}{Identification of Unobserved Heterogeneity Distributions}} \label{ivmqr:sec:identification-unobserved-distributions}

The identification result in the previous subsection relies on the characterization (\ref{ivmqr:eq:testable-representation}) of observationally equivalent structural functions.
Another fundamentally important feature of the model is the conditional distribution of the unobserved heterogeneity.
Let $\mu_{d, z}^\ast \coloneqq \cL(U \mid D = d, Z = z)$ denote the true conditional distribution of the unobserved heterogeneity $U$ given the treatment status $D = d$ and the instrument $Z = z.$
In this subsection, we establish a necessary and sufficient condition for some set $(\mu_{d, z})_{d, z}$ of distributions to be observationally equivalent to $(\mu_{d, z}^\ast)_{d, z}.$
For simplicity, we focus on the case where the treatment and instrument are binary.
The extension to cases where they take more than two values is straightforward, though the notation becomes more cumbersome.
We also assume that the conditional probability $p_{d, z} \coloneqq \bP (D = d \mid Z = z)$ is positive for all $d, z.$
When $p_{d, z} = 0$ for some $(d, z),$ the problem reduces to a trivial case, as discussed at the end of this section.
Let $\nu_{d, z} \coloneqq \cL(Y \mid D = d, Z = z)$ be the conditional distribution of the observed outcomes.
Each $\nu_{d, z}$ is directly identified from the data if $p_{d, z} > 0.$

To state our results formally, we introduce a novel notion.

\begin{definition}
Let $\mu, \tilde \mu, \nu, \tilde \nu$ be probability distributions on $\bR^p.$
The pair $(\mu, \tilde \mu)$ is \textit{simultaneously optimally coupleable} with $(\nu, \tilde \nu)$ if there exist a proper, lower semicontinuous convex function $\varphi : \bR^p \to \bR \cup \{\infty\}$ and couplings\footnote{A probability distribution $\pi$ on $\bR^p \times \bR^p$ is a coupling of $\mu$ and $\nu$ if the marginals of $\pi$ are $\mu$ and $\nu.$} $\pi, \tilde \pi$ of $(\mu, \nu), (\tilde \mu, \tilde \nu),$ respectively, such that $\supp (\pi), \supp (\tilde \pi) \subset \partial \varphi \coloneqq \{(u, y) \in \cU \times \bR^p \mid \forall v \in \cU, \ \varphi(v) \geq \varphi (u) + y^\prime (v - u)\}.$
\end{definition}

As explained in Appendix \ref{ivmqr:app:preliminaries}, the condition that $\supp (\pi) \subset \partial \varphi$ implies that $\varphi$ is an optimal transport potential from $\mu$ to $\nu.$
The simultaneously optimal coupleability requires that a \textit{single} convex function $\varphi$ serves as an optimal transport potential \textit{both} from $\mu$ to $\nu$ and from $\tilde \mu$ to $\tilde \nu.$
Furthermore, when $\mu$ and $\tilde \mu$ admit densities, the optimal transport map from $\mu$ to $\nu$ coincides with the map from $\tilde \mu$ to $\tilde \nu,$ with both being equal to $\nabla \varphi;$ that is, $\nabla \varphi$ is the multivariate quantile function of $\nu$ (resp., $\tilde \nu$) associated with $\mu$ (resp., $\tilde \mu$).

The simultaneously optimal coupleability is a critical condition for given $(\mu_{d, z})_{d, z}$ to be consistent with $(\nu_{d, z})_{d, z}$ within the model specified in Assumption \ref{ivmqr:ass:ivqr}.
It corresponds to the instrument $Z$ being excluded from the structural function, as discussed in detail below.

The following lemma provides a characterization of latent distributions observationally equivalent to the truth.

\begin{lemma} \label{ivmqr:lem:mu-observational-equivalence}
$(\tilde \mu_{d, z})_{d, z}$ is observationally equivalent to $(\mu_{d, z}^\ast)_{d, z}$ if and only if 
\begin{enumerate}
    \item for each $z \in \{0, 1\},$ $\mu = p_{0, z} \tilde \mu_{0, z} + p_{1, z} \tilde \mu_{1, z}$ holds, and
    \item for each $d \in \{0, 1\},$ $(\tilde \mu_{d, 0}, \tilde \mu_{d, 1})$ and $(\nu_{d, 0}, \nu_{d, 1})$ are simultaneously optimally coupleable.
\end{enumerate}
\end{lemma}

The necessity is straightforward.
The first condition follows from the independence between the unobserved heterogeneity $\tilde U$ and the instrument $\tilde Z,$ where the letters with tildes represent the random variables in Definition \ref{ivmqr:def:observational-equivalence} corresponding to $(\tilde \mu_{d, z})_{d, z}.$
In particular, it necessarily holds since
\begin{align*}
    \mu
    &=
    \bP (\tilde U \in \cdot \mid \tilde Z = z)
    \\
    &=
    \bP (\tilde U \in \cdot, \tilde D = 0 \mid \tilde Z = z)
    +
    \bP (\tilde U \in \cdot, \tilde D = 1 \mid \tilde Z = z)
    \\
    &=
    p_{0, z} \tilde \mu_{0, z}
    +
    p_{1, z} \tilde \mu_{1, z}
    .
\end{align*}

However, the independence is not the only restriction imposed on the model.
The conditional latent distributions $\tilde \mu_{d, z}$ are required to be consistent with the conditional distributions $\nu_{d, z},$ which are features of the observed data.
Specifically, $\nu_{d, 0}$ and $\nu_{d, 1}$ must be the push-forwards of $\tilde \mu_{d, 0}$ and $\tilde \mu_{d, 1},$ respectively, by an \textit{identical} structural function $\tilde q_d = \nabla \tilde \varphi_d,$ because
\begin{equation} \label{ivmqr:eq:push-forward-equality}
    \nu_{d, z}
    =
    \cL (\tilde Y \mid \tilde D = d, \tilde Z = z)
    =
    \cL (\nabla \tilde \varphi_d (\tilde U_d) \mid \tilde D = d, \tilde Z = z)
    =
    (\nabla \tilde \varphi_d)_\# \tilde \mu_{d, z}
\end{equation}
holds for each $(d, z).$
This is exactly what the second condition of Lemma \ref{ivmqr:lem:mu-observational-equivalence} guarantees.
From an econometric perspective, (\ref{ivmqr:eq:push-forward-equality}) claims that the structural function $\tilde q_d$ should not depend on $z;$ i.e., the instrument satisfies the exclusion restriction.

We further investigate the simultaneously optimal coupleability condition.
Define the loss function $L_d(\mu_{d, 0}, \mu_{d, 1})$ as follows:
$$
    L_d(\mu_{d, 0}, \mu_{d, 1}) 
    \coloneqq  
    W_2^2(\mu_{d, 0}, \nu_{d, 0}) 
    +
    W_2^2(\mu_{d, 1}, \nu_{d, 1})
    -
    2 W_2^2 \left(\frac{\mu_{d, 0} + \mu_{d, 1}}{2}, \frac{\nu_{d, 0} + \nu_{d, 1}}{2}\right)
$$
where $W_2 (\cdot, \cdot)$ represents the $2$-Wasserstein distance.

The loss function $L_d(\mu_{d, 0}, \mu_{d, 1})$ measures how far a given pair $(\mu_{d, 0}, \mu_{d, 1})$ is from being simultaneously optimally coupleable with $(\nu_{d, 0}, \nu_{d, 1}),$ which is formalized by the following lemma.

\begin{lemma} \label{ivmqr:lem:L-characterization}
Let $d \in \{0, 1\}$ represent a fixed treatment state. 
Then, $L_d(\mu_{d, 0}, \mu_{d, 1}) \geq 0$ holds for any $\mu_{d, 0}, \mu_{d, 1} \in \cP (\cU)$. 
Furthermore, $(\mu_{d, 0}, \mu_{d, 1})$ and $(\nu_{d, 0}, \nu_{d, 1})$ are simultaneously optimally coupleable if and only if $L_d(\mu_{d, 0}, \mu_{d, 1}) = 0.$
\end{lemma}

Combined with Lemma \ref{ivmqr:lem:mu-observational-equivalence}, this observation provides another necessary and sufficient condition for any $(\tilde \mu_{d, z})_{d, z}$ to be observationally equivalent to $(\mu_{d, z}^\ast)_{d, z}.$

\begin{theorem} \label{ivmqr:thm:observational-equivalence-mu}
$(\tilde \mu_{d, z})_{d, z} \in (\cP (\cU))^4$ is observationally equivalent to $(\mu_{d, z}^\ast)_{d, z}$ if and only if $(\tilde \mu_{d, z})_{d, z}$ solves
\begin{equation} \label{ivmqr:eq:population-L-global}
    \min_{(\mu_{d, z})_{d, z} \in (\cP (\cU))^4}
    L_0 (\mu_{0, 0}, \mu_{0, 1}) 
    + 
    L_1 (\mu_{1, 0}, \mu_{1, 1}) 
    \quad \text{s.t.} \quad 
    \begin{cases} 
        \mu = p_{0, 0} \mu_{0, 0} + p_{1, 0} \mu_{1, 0} \\ 
        \mu = p_{0, 1} \mu_{0, 1} + p_{1, 1} \mu_{1, 1} 
    \end{cases}
    .
\end{equation}
\end{theorem}

The optimization problem (\ref{ivmqr:eq:population-L-global}) can be interpreted as follows.
The problem searches for a set of latent distributions that achieves the consistency condition (\ref{ivmqr:eq:push-forward-equality}) (or equivalently, the simultaneously optimal coupleability) as much as possible.
However, the set must also respect the IV independence $\mu = p_{0, z} \mu_{0, z} + p_{1, z} \mu_{1, z},$ which appears in (\ref{ivmqr:eq:population-L-global}) as constraints.

Once we obtain $(\tilde \mu_{d, z})_{d, z}$ that is observationally equivalent to $(\mu_{d, z}^\ast)_{d, z}$ by solving (\ref{ivmqr:eq:population-L-global}), we can recover structural functions that are observationally equivalent to the true ones.
The following statement describes a way to obtain the structural functions from the latent distributions.
The proof is omitted, as the same construction of random variables used in the proof of Lemma \ref{ivmqr:lem:mu-observational-equivalence} applies.

\begin{corollary} \label{ivmqr:cor:recovery-q}
Let $(\tilde \mu_{d, z})_{d, z}$ be any solution to (\ref{ivmqr:eq:population-L-global}), and $\tilde \varphi_d$ be a convex function that simultaneously optimally couples $(\tilde \mu_{d, 0}, \tilde \mu_{d, 1})$ and $(\nu_{d, 0}, \nu_{d, 1}).$
Then, $(\tilde \varphi_0, \tilde \varphi_1)$ is observationally equivalent to $(\varphi_0^\ast, \varphi_1^\ast).$
Furthermore, if $\tilde q_d (\cdot) \in \partial \tilde \varphi_d (\cdot),$\footnote{For a convex function $f : \bR^p \to \bR \cup \{+ \infty\},$ let $\partial f (x) \coloneqq \{z \in \bR^p \mid \forall y \in \bR^p, \ f (y) \geq f (x) + z^\prime (y - x)\}$ denote its subdifferential at $x.$ See \cite{rockafellar1970convex} for more details.} then $(\tilde q_0, \tilde q_1)$ is observationally equivalent to $(q_0^\ast, q_1^\ast).$
\end{corollary}

Theorem \ref{ivmqr:thm:observational-equivalence-mu} and Corollary \ref{ivmqr:cor:recovery-q} are useful because they suggest a promising estimation procedure.
Specifically, in the next section, we provide an estimator of the latent distribution $\mu_{d, z}$ by replacing the population-level objects in (\ref{ivmqr:eq:population-L-global}) with their sample analogues.
Once we have an estimator $\hat \mu_{d, z}$ of $\mu_{d, z},$ we can construct an estimator of the structural function with the procedure described in Corollary \ref{ivmqr:cor:recovery-q}.

We close this section by discussing the cases with $p_{d, z} = 0$ for some $(d, z).$
Without loss, suppose that $p_{1, 0} = 0.$
Note first that $\mu_{1, 0}^\ast$ is not identified, as no individuals with $(D, Z) = (1, 0)$ are observed.

Recall that the complication of the identification is the conflict between the independence ($Z \indep U$) and the consistency with the data ($\nu_{d, 0} = (q_d)_\# \mu_{d, 0}$ and $\nu_{d, 1} = (q_d)_\# \mu_{d, 1}$ for an identical $q_d$).
However, when the condition $p_{1, 0} = 0$ holds, the consistency condition for $(d, z) = (1, 0)$ can be ignored because $\nu_{1, 0}$ is not identified, which eliminates this dilemma.

To see this, suppose that $(\tilde \mu_{d, z})_{d, z}$ is observationally equivalent to $(\mu_{d, z}^\ast)_{d, z}.$
Then, the identification holds straightforwardly as follows.
The IV independence implies $\tilde \mu_{0, 0} = \mu = \mu_{0, 0}^\ast.$
The consistency condition pins down the structural function $q_0^\ast$ via the equation $\nu_{0, 0} = (q_0^\ast)_\# \mu,$ which implies $\tilde \mu_{0, 1} = ((q_0^\ast)^{-1})_\# \nu_{0, 1} = \mu_{0, 1}^\ast.$
Finally, by the independence, it follows that $\tilde \mu_{1, 1} = (\mu - p_{0, 1} \tilde \mu_{0, 1}) / p_{1, 1} = (\mu - p_{0, 1} \mu_{0, 1}^\ast) / p_{1, 1} = \mu_{1, 1}^\ast.$
Therefore, $(\tilde \mu_{0, 0}, \tilde \mu_{0, 1}, \tilde \mu_{1, 1}) = (\mu_{0, 0}^\ast, \mu_{0, 1}^\ast, \mu_{1, 1}^\ast)$ holds.

This argument corresponds to solving the simpler minimization problem.
$$
    \min_{\mu_{0, 1}, \mu_{1, 1} \in (\cP (\cU))^2}
    L_0 (\mu, \mu_{0, 1})
    \quad
    \text{s.t.}
    \quad
    \mu = p_{0, 1} \mu_{0, 1} + p_{1, 1} \mu_{1, 1}
    .
$$
Compared with (\ref{ivmqr:eq:population-L-global}), $p_{1, 0} = 0$ removes one of the equality constraints, and the unidentifiability of $\nu_{1, 0}$ eliminates the second term of the objective function.

Note also that under the condition $p_{1, 0} = 0,$ the structural functions $q_0^\ast, q_1^\ast$ are globally identified.
Indeed, as discussed above, $q_0^\ast$ is simply the optimal transport map from $\mu$ to $\nu_{0, 0}.$
Furthermore, since $\mu_{1, 1}^\ast$ is identified, the function $q_1^\ast$ can then be recovered as the optimal transport map from $\mu_{1, 1}^\ast$ to $\nu_{1, 1}.$

\memo{red}{
\begin{corollary}
Suppose that $p_{1, 0} = 0.$
Under Assumptions \ref{ivmqr:ass:ivqr}, \ref{ivmqr:ass:correct-specification}, \ref{ivmqr:ass:regular-boundary}, \ref{ivmqr:ass:smooth-densities}, \ref{ivmqr:ass:div-free}, $(q_0^\ast, q_1^\ast)$ is globally identified.
\end{corollary}
}

\section{\memo{red}{Estimation}} \label{ivmqr:sec:estimation}

\subsection{The Estimator}

This section develops a practical estimation procedure to recover the structural functions. 
We restrict our attention to a binary treatment and instrument setup where $D, Z \in \{0, 1\}$. 
We abstract away from exogenous covariates $X.$
We leave the incorporation of covariates for future work.

The estimation consists of two steps.
In the first step, we estimate the true conditional probability distribution $\mu_{d, z}^\ast = \cL(U \mid D = d, Z = z)$ of the unobserved heterogeneity $U$ given the treatment status $D = d$ and instrument $Z = z.$
As in Corollary \ref{ivmqr:cor:recovery-q}, the structural function $q_d^\ast$ can be recovered from an estimate of $\mu_{d, z}^\ast$ and that of the conditional distribution $\nu_{d, z} = \cL(Y \mid D = d, Z = z)$ of the observed outcomes.

We assume the availability of an independent and identically distributed sample $(Y_i, D_i, Z_i)$ for $i = 1, \dots, n.$
For simplicity, we explain the estimation procedure by focusing mainly on cases where $p_{d, z} = \bP (D = d \mid Z = z) > 0$ for all $d, z.$
However, our theory still applies to cases where $p_{d, z} = 0$ for some $d, z$ with only minor modifications.
See Remark \ref{ivmqr:rem:p10=0} for these adjustments.

Let $\hat \nu_{d, z}$ be an estimator of $\nu_{d, z}$ that is strongly consistent in the $2$-Wasserstein distance, i.e., $W_2(\hat{\nu}_{d, z}, \nu_{d, z}) \to 0$ almost surely.
This convergence assumption is satisfied by the empirical measure
$$
    \hat \nu_{d, z} (\cdot)
    =
    \dfrac{\sum_{i = 1}^n \bI\{D_i = d, Z_i = z\} \delta_{Y_i} (\cdot)}{\sum_{i = 1}^n \bI\{D_i = d, Z_i = z\}}
    ,
$$
where the term $\delta_{Y_i}$ is the Dirac point mass at the observation $Y_i.$
This follows from Theorem 6.9 of \cite{villani2009optimal} and Theorem 11.4.1 of \cite{dudley2002real}.

Similarly, we require a strongly consistent estimator $\hat{p}_{d, z}$ that almost surely converges to the true conditional probability $p_{d, z}.$
In what follows, we use
$$
    \hat p_{d, z}
    =
    \dfrac{\sum_{i = 1}^n \bI\{D_i = d, Z_i = z\}}{\sum_{i = 1}^n \bI\{Z_i = z\}}
    ,
$$
although other choices are possible.

We replace all population-level objects in the loss function $L_d$ with their respective sample counterparts. 
Define the empirical loss function as:
$$
    \hat{L}_d(\mu_{d, 0}, \mu_{d, 1}) 
    \coloneqq 
    W_2^2(\mu_{d, 0}, \hat{\nu}_{d, 0}) 
    +
    W_2^2(\mu_{d, 1}, \hat{\nu}_{d, 1})
    -
    2 W_2^2 \left(\frac{\mu_{d, 0} + \mu_{d, 1}}{2}, \frac{\hat{\nu}_{d, 0} + \hat{\nu}_{d, 1}}{2}\right) 
    .
$$

We then define the vector $(\hat{\mu}_{0, 0}, \hat{\mu}_{0, 1}, \hat{\mu}_{1, 0}, \hat{\mu}_{1, 1})$ to be any minimizer of
\begin{equation} \label{ivmqr:eq:empirical-L}
    \min_{(\mu_{d, z})_{d, z} \in \cM}
    \hat L_0 (\mu_{0, 0}, \mu_{0, 1}) 
    + 
    \hat L_1 (\mu_{1, 0}, \mu_{1, 1}) 
    \quad \text{s.t.} \quad 
    \begin{cases} 
        \mu = \hat p_{0, 0} \mu_{0, 0} + \hat p_{1, 0} \mu_{1, 0} \\ 
        \mu = \hat p_{0, 1} \mu_{0, 1} + \hat p_{1, 1} \mu_{1, 1} 
    \end{cases}
    ,
\end{equation}
where $\cM \subset (\cP (\cU))^4$ is to be determined next.

Recall that any $(\tilde \mu_{d, z})_{d, z}$ observationally equivalent to the truth solves the population problem (\ref{ivmqr:eq:population-L-global}).
If global identification fails, and therefore, (\ref{ivmqr:eq:population-L-global}) admits multiple solutions, any estimator is inconsistent (as a point estimator).

To avoid this difficulty, we shrink the variable space so that global identification holds.
Specifically, we consider a set $\cM \subset (\cP (\cU))^4$ with $(\mu_{d, z}^\ast)_{d, z} \in \cM$ such that it is the unique solution to
\begin{equation} \label{ivmqr:eq:population-L}
    \min_{(\mu_{d, z})_{d, z} \in \cM}
    L_0 (\mu_{0, 0}, \mu_{0, 1}) 
    + 
    L_1 (\mu_{1, 0}, \mu_{1, 1}) 
    \quad \text{s.t.} \quad 
    \begin{cases} 
        \mu = p_{0, 0} \mu_{0, 0} + p_{1, 0} \mu_{1, 0} \\ 
        \mu = p_{0, 1} \mu_{0, 1} + p_{1, 1} \mu_{1, 1} 
    \end{cases}
    .
\end{equation}
Note that in the literature on estimation in locally identified models, it is common to restrict the parameter space so that global identification holds.
See, for example, \cite{chernozhukov2007instrumental}.

\begin{remark} \label{ivmqr:rem:p10=0}
When $\hat p_{1, 0} = 0,$ it is enough to solve the simpler problem
\begin{equation} \label{ivmqr:eq:empirical-L-zero}
    \min_{\mu_{0, 1}, \mu_{1, 1}}
    \hat L_0 (\mu, \mu_{0, 1})
    \quad \text{s.t.} \quad 
    \mu = \hat p_{0, 1} \mu_{0, 1} + \hat p_{1, 1} \mu_{1, 1} 
\end{equation}
instead of (\ref{ivmqr:eq:empirical-L}).
The cases with $\hat p_{0, 1} = 0$ can be handled in a similar way.
\end{remark}

Given minimizers of (\ref{ivmqr:eq:empirical-L}) (or (\ref{ivmqr:eq:empirical-L-zero})), we proceed to the estimation of structural functions.
For each $d \in \{0, 1\},$ let $\hat{\varphi}_d : \cU \to \bR$ be an optimal transport potential, which is defined in Appendix \ref{ivmqr:app:preliminaries},
$$
    \text{from}
    \quad
    \begin{cases}
        (\hat{\mu}_{d, 0} + \hat{\mu}_{d, 1}) / 2 & (\hat p_{d, 1 - d} \neq 0) \\
        \hat{\mu}_{d, d} & (\hat p_{d, 1 - d} = 0)
    \end{cases}
    \quad
    \text{to}
    \quad
    \begin{cases}
        (\hat{\nu}_{d, 0} + \hat{\nu}_{d, 1}) / 2 & (\hat p_{d, 1 - d} \neq 0) \\
        \hat{\nu}_{d, d} & (\hat p_{d, 1 - d} = 0)
    \end{cases}
    .
$$
We then estimate the structural function $q_d^\ast$ with any selection $\hat{q}_d (\cdot)$ of the subdifferential of $\hat \varphi_d;$ that is, $\hat{q}_d (\cdot) \in \partial \hat{\varphi}_d (\cdot).$

Our estimators $(\hat \mu_{d, z})_{d, z}$ and $\hat{q}_d$ are consistent in the following sense.
This result holds even when $p_{d, z} = 0$ for some $d, z.$

\begin{theorem} \label{ivmqr:thm:q-uniform-consistency}
(i) 
Assume that $W_2 (\hat \nu_{d, z}, \nu_{d, z}) \to 0$ and $\hat p_{d, z} \to p_{d, z}$ hold almost surely for each $(d, z)$ with $p_{d, z} > 0$, that $(\mu_{d, z}^\ast) \in \cM$ is the unique solution to (\ref{ivmqr:eq:population-L}), that $\bP (D = 0, Z = 0), \bP (D = 1, Z = 1) > 0$ and that $\cM$ is compact in the $2$-Wasserstein distance and contains a small ball around $(\mu_{d, z}^\ast).$
Then, the estimator $\hat \mu_{d, z}$ is strongly consistent if $p_{d, z} > 0;$ that is, $W_2 (\hat \mu_{d, z}, \mu_{d, z}^\ast) \to 0$ almost surely. 
(ii) If, in addition, $\varphi_d^\ast$ is differentiable on $\Int(\supp (U \mid D = d)),$ then the estimator $\hat{q}_d$ is uniformly consistent on any compact set $K \subset \Int(\supp (U \mid D = d)),$ i.e.,
$$
    \lim_{n \to \infty} 
    \sup_{u \in K} 
    \norm{\hat{q}_d(u) - q_d^\ast(u)} 
    = 
    0 
    \quad 
    \text{a.s.}
    .
$$
\end{theorem}

To implement the proposed estimator in practice, one needs to set $\cM$ so that $(\mu_{d, z}^\ast) \in \cM$ and (\ref{ivmqr:eq:population-L}) has a unique solution, which is challenging when identification is local.
If $\cM$ is specified too broadly such that it contains multiple solutions, the estimator could become unstable.
In simulation studies in Section \ref{ivmqr:sec:simulation-studies}, we nevertheless set $\cM = (\cP (\cU))^4.$
Fortunately, the estimator is stable and approximates the true structural function well.
This suggests the possibility that the structural function is globally identified under certain additional assumptions.

\subsection{Simulation Studies} \label{ivmqr:sec:simulation-studies}

This subsection evaluates the finite-sample performance of the proposed estimator using Monte Carlo simulations. 

The unobserved heterogeneity vector $U = (U^0, U^1)^\prime$ is drawn uniformly from $\cU = [0, 1]^2.$
The potential outcomes are generated as $Y_d = q_d^\ast (U)$, with the structural functions $q_d^\ast$ specified below. 
For the baseline $d = 0$, we set $q_0^\ast (u) = u$.
For the treated state $d = 1$, the structural function takes the following nonlinear form:
$$
    q_1^\ast (u) 
    =
    \frac{1}{3.5}
    \begin{pmatrix}
        \frac{e^{u^0} - 1}{e - 1}
        +
        (u^0)^2
        +
        u^0 u^1
        +
        \frac{1}{2} 
        (u^1)^2
        \\
        \frac{e^{u^1} - 1}{e - 1}
        +
        \frac{1}{2}
        (u^0)^2
        +
        u^0 u^1
        +
        (u^1)^2
    \end{pmatrix}
    .
$$
By design, both $q_0^\ast$ and $q_1^\ast$ possess symmetric positive definite derivatives, and therefore, they are the gradients of convex functions.
Figure \ref{ivmqr:fig:contour-plot-true-row} displays contour plots of the true structural functions $q_0^\ast$ and $q_1^\ast$. 

\begin{figure}[ht]
    \centering
    \includegraphics[width=1\linewidth]{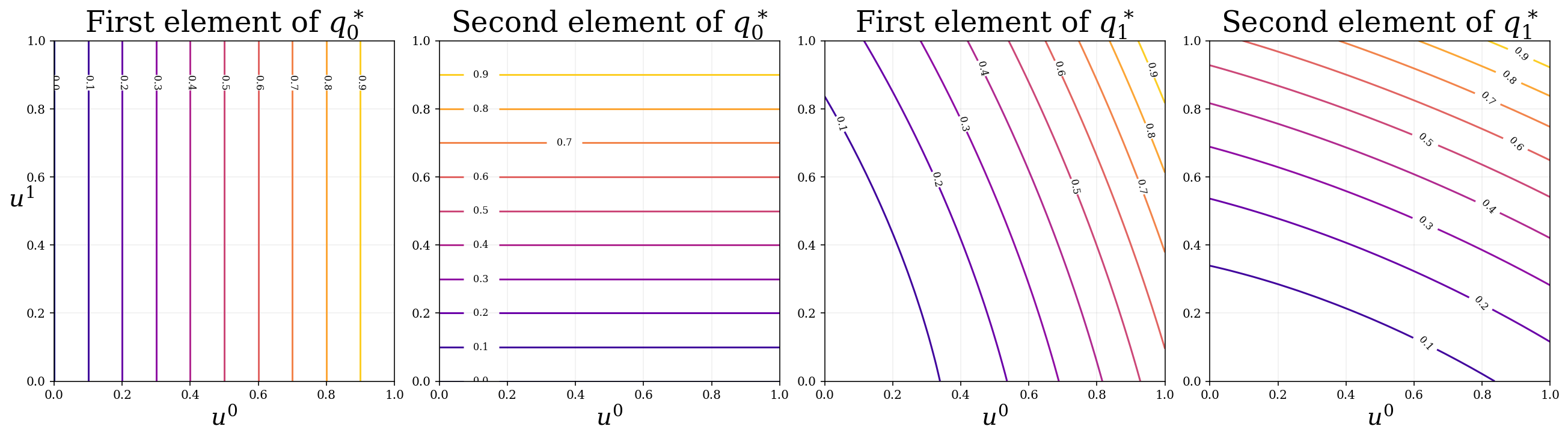}
    \caption{Contour plots of $q_0^\ast$ and $q_1^\ast.$ In each panel, the $x$-axis represents the first element $u^0$ of the unobserved heterogeneity vector $u$, and the $y$-axis is the second element $u^1.$ In the left panel, each line connects values of $u$ at which the first element of $q_0^\ast$ takes the same value. The other three panels depict similar level curves for the second element of $q_0^\ast$ and the first and second elements of $q_1^\ast.$}
    \label{ivmqr:fig:contour-plot-true-row}
\end{figure}

To model treatment selection, we draw an independent noise vector $V = (V^0, V^1)^\prime$ uniformly from $[-0.7, 0.7]^2$.
The binary treatment indicator is assigned via $D = \argmax \{U^0 + V^0, U^1 + V^1\}$.
The instrumental variable is generated as $Z = \argmax \{V^0, V^1\}$. 
This construction ensures that $Z$ is relevant for $D$ while maintaining $U \indep Z$.
For the estimation, we simulate an i.i.d. sample of $n = 3000$ tuples $(Y, D, Z)$.

We solve the empirical minimization problem (\ref{ivmqr:eq:empirical-L}) to recover the estimates $\hat{q}_d$.
We set the search space to $\cM = (\cP (\cU))^4$, assuming global identification.
To render this infinite-dimensional minimization tractable, we discretize the support $\cU$.
Consequently, the marginal $\mu$ and each conditional distribution $\tilde \mu_{d, z}$ are approximated as discrete distributions over this common grid.
While imposing rich parametric restrictions on $\cM$ offers an alternative computational approach, we do not pursue it here.

To solve the discretized minimization problem, we use Sequential Least Squares Programming (SLSQP).
See \cite{kraft1988software} for details of the algorithm.

Figure \ref{ivmqr:fig:contour-plot-estimated-row} presents the contour plots for the estimated functions $\hat q_0$ and $\hat q_1$. 
Comparing these to Figure \ref{ivmqr:fig:contour-plot-true-row} visually confirms the estimator's ability to accurately recover the underlying structural shapes.

\begin{figure}[ht]
    \centering
    \includegraphics[width=1\linewidth]{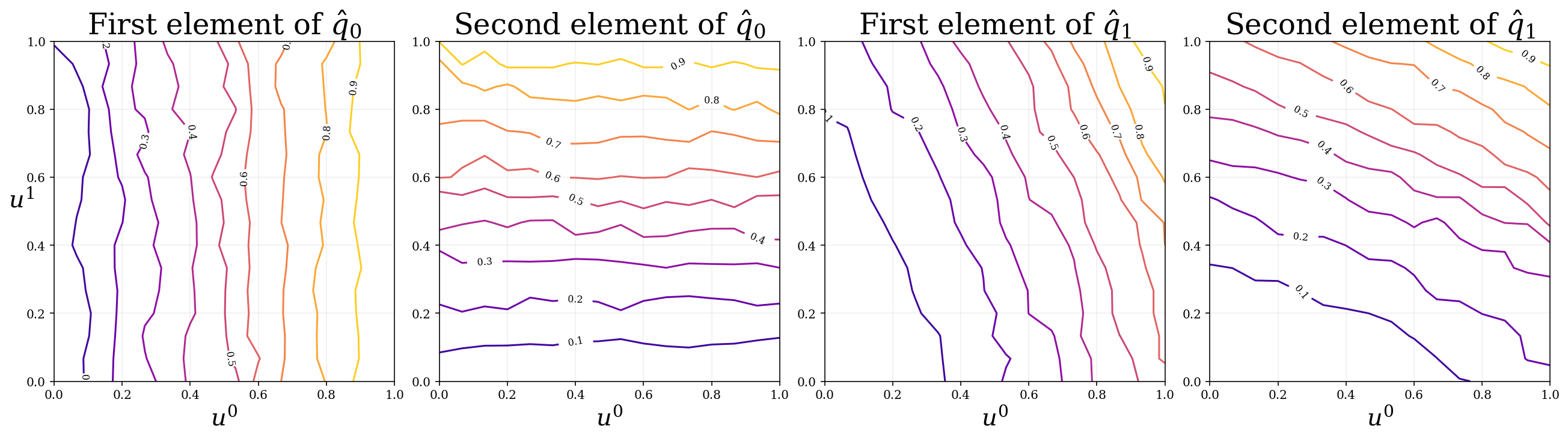}
    \caption{Contour plots of $\hat q_0$ and $\hat q_1.$}
    \label{ivmqr:fig:contour-plot-estimated-row}
\end{figure}

Since the problem (\ref{ivmqr:eq:empirical-L}) is non-convex, solutions returned by SLSQP may depend on the initial values.
To find the global minimum, we try several initial values and adopt the one that produces the smallest objective value.
A key observation is that solutions are close to the truth regardless of initial values.
This brings us two insights.
Firstly, SLSQP successfully finds the global minimum even though the problem is non-convex.
Moreover, it is computationally supported that the true structural function $q_d^\ast$ is the only one consistent with the observed data; that is, it is globally identified in the current data generating process.

\section{Concluding Remarks} \label{ivmqr:sec:conclusion}

In this paper, we develop a novel nonseparable instrumental variable model to address endogeneity in settings with multivariate outcomes.
Recognizing that estimating separate one-dimensional quantile models often invalidates fundamental assumptions like rank similarity, our framework models the joint distribution of the outcomes.
By restricting the structural function to be the gradient of a convex function, our model accommodates a flexible correlation structure among the outcome dimensions.
Under this model, we establish conditions for the local identification of the structural functions when treatments and instruments are discrete.
We demonstrate that this identification relies on a completeness-type condition, which generalizes the standard full-rank condition used in one-dimensional IVQR models.
We also propose a computationally tractable minimum-distance estimator and establish its uniform consistency.
Finally, simulation studies confirm that our proposed estimator performs stably in finite samples and provide computational evidence supporting the potential for global identification under additional assumptions.

\appendix

\section{Preliminaries} \label{ivmqr:app:preliminaries}

\subsection{Optimal transport and multivariate quantiles}

In this section, we briefly go over mathematical concepts that are used in the paper.
Fix a probability measure $\mu$ on a compact convex subset $\cU$ of the $p$-dimensional Euclidean space, and let $Y$ be a $p$-dimensional random vector.
First, consider the case of $p = 1$ and $\mu = U [0, 1].$
According to the inverse probability integral transform, there exists a random variable $U \sim U [0, 1]$ such that $Y = q (U)$ almost surely where $q$ is the quantile function of $Y$ (see, for example, Proposition 3.2 of \cite{shorack2000probability}).
Moreover, $q$ is the (almost surely) unique nondecreasing function that satisfies this equation for some uniform random variable.
This observation for the classical one-dimensional case holds more generally by replacing the monotonicity of $q$ with the restriction that it is written as the gradient of a convex function.

To define multivariate quantiles developed by \cite{chernozhukov2017monge}, \cite{hallin2021distribution}, and \cite{ghosal2022multivariate}, we first review the fundamental result of optimal transport theory.
For two probability distributions $\mu, \nu,$ let $\Pi (\mu, \nu)$ be the set of couplings of $\mu$ and $\nu;$ that is, it is the set of joint distributions whose marginals are $\mu$ and $\nu.$

\begin{theorem}[Theorem 2.12, \cite{villani2003topics}] \label{ivmqr:thm:ot-fundamental}
Let $\mu, \nu$ be probability measures on $\bR^p$ with finite second moments. 
(1)
A coupling $\pi^\ast \in \Pi (\mu, \nu)$ solves 
\begin{equation} \label{ivmqr:eq:correlation-maximization}
    \sup_{\pi \in \Pi (\mu, \nu)} \int u^\prime y d \pi
\end{equation}
if and only if there exists a convex lower semicontinuous function $\varphi$ such that 
\begin{equation} \label{ivmqr:eq:supp-in-subdifferential}
    \supp (\pi^\ast)
    \subset
    \partial \varphi
    \coloneqq
    \{(u, y) \in \cU \times \bR^p \mid \forall v \in \cU, \ \varphi(v) \geq \varphi (u) + y^\prime (v - u)\}
    .
\end{equation}
(2)
If $\mu$ is absolutely continuous with respect to the Lebesgue measure on $\bR^p,$ then there exists a unique optimal $\pi^\ast,$ and it satisfies $\pi^\ast = (\id, \nabla \varphi)_\# \mu.$
In particular, $\nabla \varphi$ is unique $\mu$-almost surely, and $\nu = (\nabla \varphi)_\# \mu$ holds.
\end{theorem}

From the probabilistic perspective, the problem (\ref{ivmqr:eq:correlation-maximization}) maximizes the correlation between two random variables $U \sim \mu$ and $Y \sim \nu.$
Clearly, this is equivalent to the problem
$$
    \inf_{\pi \in \Pi (\mu, \nu)} \int \norm{u - y}^2 d \pi
    ,
$$
which searches for the cheapest way to ``transport'' the mass from the source distribution $\mu$ to the target distribution $\nu.$
Consequently, any solution $\pi^\ast$ to $(\ref{ivmqr:eq:correlation-maximization})$ is referred to as an \textit{optimal transport plan} from $\mu$ to $\nu.$

The optimality of a given transport plan $\pi^\ast$ is characterized by its support being contained within the graph of the subdifferential $\partial \varphi$ of a convex function $\varphi,$ as stated in (\ref{ivmqr:eq:supp-in-subdifferential}).
Therefore, to show that some $\pi^\ast$ is optimal, it suffices to find a convex function satisfying (\ref{ivmqr:eq:supp-in-subdifferential}).
The convex function $\varphi$ provides an alternative way to describe optimal transport plans, and it is called an \textit{optimal transport potential}.

Moreover, whenever the source measure $\mu$ is absolutely continuous, the optimal transport plan is unique.
Under this unique plan, the mass at $u$ is deterministically transported to $\nabla \varphi (u),$ where $\varphi$ is an optimal transport potential.
The map $\nabla \varphi$ is called the \textit{optimal transport map}.

Now, we are ready to define multivariate quantiles.
Let $\mu$ be an absolutely continuous reference distribution.
With $\mu$ fixed, the distribution $\cL (Y)$ of $Y$ is characterized by the optimal transport map from $\mu$ to $\cL (Y).$
In other words, for each distribution $\nu = \cL (Y),$ there uniquely exists a function $q,$ which has the form of $q = \nabla \varphi$ for some convex function $\varphi,$ such that $\nu = q_\# \mu.$
The function $q = \nabla \varphi$ is called the \textit{multivariate quantile function} of $Y$ (associated with $\mu$).

\memo{red}{
Recalling Definition \ref{ivmqr:def:cyclical-monotonicity} and Theorem \ref{ivmqr:thm:rockafellar}, multivariate quantile functions are cyclically monotone.
As the following theorem shows, if a multivariate quantile function satisfies regularity conditions, it is homeomorphic on the interior of its domain.
\begin{theorem} \label{ivmqr:thm:ot-map-is-homeo}
Suppose that a continuous function $q : \cU \to \bR^p$ is strictly cyclically monotone.
Also, assume that its image $\cY \coloneqq q (\cU)$ is convex.
Then, $q |_{\Int (\cU)}$ is a homeomorphism between $\Int (\cU)$ and $\Int (\cY).$
\end{theorem}
}

In the proof of Theorem \ref{ivmqr:thm:q-uniform-consistency}, the convergence of optimal transport potentials plays a central role.
The following result is a localized version of Theorem 2.8 in \cite{del2019central}. 
Its proof is omitted, as the argument for the original theorem applies similarly.
\begin{theorem} \label{ivmqr:thm:del-barrio}
Assume $\mu_n, \mu, \nu_n, \nu$ are probabilities such that $W_2 (\mu_n, \mu) \to 0$ and $W_2 (\nu_n, \nu) \to 0,$ and that $\mu$ admits a density.
Let $D \subset \Int (\supp (\mu))$ be open and convex.
If $\psi_n$ (resp., $\psi$) are optimal transportation potentials from $\mu_n$ to $\nu_n$ (resp., from $\mu$ to $\nu$) then there exist constants $a_n$ such that if $\tilde \psi_n = \psi_n - a_n$ then $\tilde \psi_n (x) \to \psi (x)$ for every $x \in D.$
\end{theorem}

\begin{proposition} \label{ivmqr:prop:consistency-ot-map}
Let $\mu_n, \nu_n, \mu, \nu \in \cP (\bR^d)$ be probability measures such that $\mu$ has a density.
Assume $W_2 (\mu_n, \mu) \to 0$ and $W_2 (\nu_n, \nu) \to 0.$
Let $\varphi_n$ (resp., $\varphi$) be an optimal transport potential from $\mu_n$ (resp., $\mu$) to $\nu_n$ (resp., $\nu$).
If $\varphi$ is differentiable on a convex open set $D \subset \Int (\supp (\mu)),$ then the optimal transport map $T$ from $\mu$ to $\nu$ satisfies 
$$
    \lim_{n \to \infty}
    \sup_{x \in K}
    \norm{T_n (x) - T (x)}
    =
    0
$$
for any sequence $T_n (x) \in \partial \varphi_n (x)$ and compact set $K \subset D.$
\end{proposition}

\subsection{Piecewise $C^k$ Boundary}
\begin{definition}
Let $K\subset\mathbb R^p$ be a compact set with nonempty interior. 
We say that $K$ has a \textit{piecewise $C^k$ boundary} if there exist an open neighborhood $V$ of $K$ and finitely many functions $\rho_1,\ldots,\rho_J\in C^k(V)$ such that
\begin{gather*}
    \Int (K)
    = \{x \in V \mid \rho_j(x) < 0,\ j = 1,\ldots,J\}
    ,
    \\
    K
    = \{x \in V \mid \rho_j(x) \leq 0,\ j = 1,\ldots,J\}
    .
\end{gather*}
and, for every $x\in\partial K,$ the collection
$$ 
    \{
        \nabla \rho_j (x) 
        \mid 
        \rho_j (x) =0
    \} 
$$
is linearly independent.
\end{definition}

To the best of the author's knowledge, this definition was introduced by \cite{pinchuk1980holomorphic}.
It is general enough to cover smooth boundaries, rectangles and boxes.

\section{Local identification for system of equations} \label{ivmqr:app:local-identification-system}

\subsection{Setup and theorem}

The proof of Theorem \ref{ivmqr:thm:local-identification-IVQR-binary} relies on a more general identification result (Theorem \ref{ivmqr:thm:local-identification-general}), which is of independent interest.
We devote this section to the identification result that can be applied to parameters that are characterized by the solution of a system of equations.

Let $(\cA, \norm{\cdot}_\cA), (\cB, \norm{\cdot}_\cB)$ be normed spaces, $\cA_0 \subset \cA$ be any subset, and $\cZ$ be a finite set.
For $z \in \cZ,$ fix $\phi_z : \cA_0 \to \cB$ and $b_z \in \cB.$
We consider the following system of equations
\begin{align} \label{ivmqr:eq:simultaneous-equations}
    \phi_z (a) 
    = 
    b_z
    \ \text{for} \ 
    z \in \cZ
    .
\end{align}
Suppose $a = a^\ast \in \cA_0$ is a solution to (\ref{ivmqr:eq:simultaneous-equations}).
We are interested in whether $a = a^\ast$ is the locally unique solution.
To formulate deviations from $a^\ast,$ we define the tangent set at $a = a^\ast$ as
\begin{align*}
    T_{a^\ast}
    \coloneqq
    \{
        \delta a \in \cA 
        \mid 
        a^\ast + \delta a \in \cA_0
    \}
    .
\end{align*}
The set $T_{a^\ast}$ specifies possible deviations in $\cA_0$ from $a^\ast.$
Notice that if $\cA_0 = \cA,$ then $T_{a^\ast} = \cA.$
We also define the set of normalized tangent vectors as
\begin{align*}
    T_{a^\ast}^1
    \coloneqq
    \left \{
        \frac{\delta a}{\norm{\delta a}_\cA}
        \mid
        \delta a \in T_{a^\ast} \setminus \{0\}
    \right \}
    .
\end{align*}

Remember that when $\cA$ and $\cB$ are finite-dimensional, the locally unique solvability of a system of equations is given by the implicit function theorem, which assumes the differentiability and the full rank of the Jacobian matrix.
As our identification result is an analogous statement for infinite-dimensional parameters, we impose assumptions corresponding to these two.

The first assumption is the differentiability of $\phi_z$ at $a = a^\ast.$
\begin{assumption} \label{ivmqr:ass:differentiable}
For $z \in \cZ,$ there exists $(\phi_z)_{a^\ast}^\prime: T_{a^\ast}^1 \to \cB$ such that
\begin{align*}
    \lim_{\varepsilon \downarrow 0}
    \sup_{
        \substack{
            h \in T_{a^\ast}^1
            \\
            a^\ast + \varepsilon h \in \cA_0
        }
    }
    \norm{\frac{\phi_z (a^\ast + \varepsilon h) - \phi_z (a^\ast)}{\varepsilon} - (\phi_z)_{a^\ast}^\prime (h)}_\cB
    =
    0
\end{align*}
where if the set over which the supremum is taken is empty, the value reads zero.
\end{assumption}

It is clear that if $\cA_0 = \cA,$ Assumption \ref{ivmqr:ass:differentiable} reduces to the standard Fr\'echet differentiability.
For strict subsets of $\cA,$ it is weaker than the Fr\'echet differentiability.
In particular, even if $\cA_0$ is not a vector space, Assumption \ref{ivmqr:ass:differentiable} is well-defined.
Unlike the Fr\'echet differentiability, the linearity of $(\phi_z)_{a^\ast}^\prime$ is not required.

The next assumption is an infinite-dimensional version of the full rank of the Jacobian matrix.
\begin{assumption} \label{ivmqr:ass:full-rank}
It holds that
\begin{align*}
    \inf_{h \in T_{a^\ast}^1}
    \sum_{z \in \cZ}
    \norm{(\phi_z)_{a^\ast}^\prime (h)}_\cB
    >
    0
    .
\end{align*}
\end{assumption}

Now, we are ready to state the main theorem in this section.
\begin{theorem} \label{ivmqr:thm:local-identification-general}
Under Assumptions \ref{ivmqr:ass:differentiable} and \ref{ivmqr:ass:full-rank}, $a^\ast$ is the locally unique solution to (\ref{ivmqr:eq:simultaneous-equations}).
That is, there exists $\varepsilon > 0$ such that for all $a \in \cA_0,$ if $0 < \norm{a - a^\ast}_\cA < \varepsilon,$ then $a$ is not a solution to the system.
\end{theorem}

This result is closely related to Section 2 of \cite{chen2014local}.
Indeed, their Theorem 1 can be shown from Theorem \ref{ivmqr:thm:local-identification-general}.
Our theorem has a different focus than their result in that it covers the local identification problem for systems of equations and that the domain $\cA_0$ of functionals is not assumed to be a Banach space.

\subsection{Proof of Theorem \ref{ivmqr:thm:local-identification-general}}

We use the following lemma that gives an equivalent representation of Assumption \ref{ivmqr:ass:full-rank}.
\begin{lemma} \label{ivmqr:lem:full-rank-version}
Let $z_0 \in \cZ.$
Assumption \ref{ivmqr:ass:full-rank} holds if and only if there exists $\eta > 0$ such that for all $h \in T_{a^\ast}^1,$
\begin{align*}
    \sum_{z \in \cZ \setminus \{z_0\}}
    \norm{
        (\phi_z)_{a^\ast}^\prime (h)
    }_\cB 
    < 
    \eta
    \Rightarrow
    \norm{(\phi_{z_0})_{a^\ast}^\prime (h)}_\cB 
    \geq 
    \eta
    .
\end{align*}
\end{lemma}

\begin{proof}[Proof of Theorem \ref{ivmqr:thm:local-identification-general}]
Take $\eta > 0$ in Lemma \ref{ivmqr:lem:full-rank-version}.
By Assumption \ref{ivmqr:ass:differentiable}, there is $\varepsilon > 0$ such that for all $\varepsilon^\prime \in (0, \varepsilon)$ and $z \in \cZ,$
\begin{align*}
    \sup_{
        \substack{
            h \in T_{a^\ast}^1
            \\
            a^\ast + \varepsilon^\prime h \in \cA_0
        }
    }
    \norm{\frac{\phi_z (a^\ast + \varepsilon^\prime h) - \phi_z (a^\ast)}{\varepsilon^\prime} - (\phi_z)_{a^\ast}^\prime (h)}_\cB
    <
    \frac{\eta}{|\cZ|}
    .
\end{align*}
Fix $z_0 \in \cZ,$ and suppose that $a \in \cA_0$ with $0 < \norm{a - a^\ast} < \varepsilon$ satisfies $\phi_z (a) = b_z$ for $z \in \cZ \setminus \{z_0\}.$
We will show $\phi_{z_0} (a) \neq b_{z_0}.$
By Assumption \ref{ivmqr:ass:differentiable}, 
\begin{align*}
    &\phantom{{}={}}
    \sum_{z \in \cZ \setminus \{z_0\}}
    \norm{
        (\phi_z)_{a^\ast}^\prime \left(\frac{a - a^\ast}{\norm{a - a^\ast}_\cA}\right)
    }_\cB
    \\
    &=
    \sum_{z \in \cZ \setminus \{z_0\}}
    \norm{
        \frac{\phi_z (a) - \phi_z (a^\ast)}{\norm{a - a^\ast}_\cA}
        -
        (\phi_z)_{a^\ast}^\prime \left(\frac{a - a^\ast}{\norm{a - a^\ast}_\cA}\right)
    }_\cB
    \\
    &\leq
    \sup_{
        \substack{
            h \in T_{a^\ast}^1
            \\
            a^\ast + \norm{a - a^\ast}_\cA h \in \cA_0
        }
    }
    \sum_{z \in \cZ \setminus \{z_0\}}
    \norm{
        \frac{\phi_z (a^\ast + \norm{a - a^\ast}_\cA h) - \phi_z (a^\ast)}{\norm{a - a^\ast}_\cA}
        -
        (\phi_z)_{a^\ast}^\prime (h)
    }_\cB
    \\
    &<
    \eta
    ,
\end{align*}
where the equality holds because $\phi_z (a) = b_z = \phi_z (a^\ast),$ the first inequality holds because $\frac{a - a^\ast}{\norm{a - a^\ast}_\cA} \in T_{a^\ast}^1,$ and the last one holds because $\norm{a - a^\ast}_\cA < \varepsilon.$
By Assumption \ref{ivmqr:ass:full-rank} and Lemma \ref{ivmqr:lem:full-rank-version}, $\norm{(\phi_{z_0})_{a^\ast}^\prime \left(\frac{a - a^\ast}{\norm{a - a^\ast}_\cA}\right)}_\cB \geq \eta.$
Therefore,
\begin{align*}
    &\phantom{{}={}}
    \norm{\phi_{z_0} (a) - b_{z_0}}_\cB
    \\
    &=
    \norm{\phi_{z_0} (a) - \phi_{z_0} (a^\ast)}_\cB
    \\
    &=
    \norm{
        \left(
            \phi_{z_0} (a) 
            - 
            \phi_{z_0} (a^\ast)
            - 
            \norm{a - a^\ast}_\cA (\phi_{z_0})_{a^\ast}^\prime 
            \left(\frac{a - a^\ast}{\norm{a - a^\ast}_\cA}\right)
        \right)
        +
        \norm{a - a^\ast}_\cA 
        (\phi_{z_0})_{a^\ast}^\prime 
        \left(\frac{a - a^\ast}{\norm{a - a^\ast}_\cA}\right)
    }_\cB
    \\
    &\geq
    \norm{a - a^\ast}_\cA
    \left(
        \underbrace{
            \norm{
                (\phi_{z_0})_{a^\ast}^\prime 
                \left(\frac{a - a^\ast}{\norm{a - a^\ast}_\cA}\right)
            }
        }_{\geq \eta}
        -
        \underbrace{
            \norm{
                \frac{
                    \phi_{z_0} (a) 
                    - 
                    \phi_{z_0} (a^\ast)
                }{\norm{a - a^\ast}_\cA}
                - 
                (\phi_{z_0})_{a^\ast}^\prime 
                \left(\frac{a - a^\ast}{\norm{a - a^\ast}_\cA}\right)
            }
        }_{< \eta / |\cZ|}
    \right)
    \\
    &>
    0
\end{align*}
which immediately implies $\phi_{z_0} (a) \neq b_{z_0}.$
\end{proof}

\section{Proofs} \label{ivmqr:app:proofs}

\subsection{Proof of Proposition \ref{ivmqr:prop:IV-independence}}

\begin{proof}
Let $d_0 \in \cD.$
Then, it follows that
\begin{align*}
    \bP (U \in \cdot \mid Z, X)
    &=
    \bE [\bP (U_D \in \cdot \mid \nu, Z, X) \mid Z, X]
    \\
    &=
    \sum_{d \in \cD}
    \bE \left[
        \bP (U_d \in \cdot \mid \nu, Z, X)
        \bI \{D = d\}
        \mid 
        Z, X
    \right]
    \\
    &=
    \sum_{d \in \cD}
    \bE \left[
        \bP (U_{d_0} \in \cdot \mid \nu, Z, X)
        \bI \{D = d\}
        \mid 
        Z, X
    \right]
    \\
    &=
    \bE \left[
        \bP (U_{d_0} \in \cdot \mid \nu, Z, X)
        \mid 
        Z, X
    \right]
    \\
    &=
    \bP (U_{d_0} \in \cdot \mid Z, X)
    \\
    &=
    \bP (U_{d_0} \in \cdot \mid X)
    \\
    &=
    \mu (\cdot)
    ,
\end{align*}
where the third equality holds by (\ref{ivmqr:ass:ivqr-A3}), the sixth by (\ref{ivmqr:ass:ivqr-A2}) and the last by (\ref{ivmqr:ass:ivqr-A1}).
\end{proof}

\subsection{Proof of Theorem \ref{ivmqr:thm:representation}}
\begin{proof}
Suppose that $q$ is observationally equivalent to $q^\ast.$
Let $B \subset \cU$ be a measurable subset, and fix some $d_0 \in \cD.$
Since $\tilde Y_d = q_d (\tilde U_d, \tilde X),$ it holds by (\ref{ivmqr:ass:ivqr-A1}) and Theorem \ref{ivmqr:thm:ot-map-is-homeo} that
\begin{align*}
    \bP (Y \in q_D (B, X) \mid X, Z) 
    &=
    \sum_{d \in \cD}
    \bP (\tilde Y_d \in q_d (B, \tilde X), \tilde D = d \mid \tilde X, \tilde Z) 
    \\
    &=
    \sum_{d \in \cD}
    \bP (\tilde U_d \in B, \tilde D = d \mid \tilde X, \tilde Z) 
    \\
    &= 
    \bP (\tilde U \in B \mid \tilde X, \tilde Z)
    .
\end{align*}
By (\ref{ivmqr:eq:structural-equation}) and (\ref{ivmqr:ass:ivqr-A3}), we have 
$$
    \bP (\tilde U \in B \mid \tilde X, \tilde Z) 
    = 
    \bE [\bP (\tilde U_{d_0} \in B \mid \tilde X, \tilde Z, \tilde \nu) \mid \tilde X, \tilde Z] 
    = 
    \bP (\tilde U_{d_0} \in B \mid \tilde X, \tilde Z)
    .
$$
We also have $\bP (\tilde U_{d_0} \in B \mid \tilde X, \tilde Z) = \bP (\tilde U_{d_0} \in B \mid \tilde X)$ by (\ref{ivmqr:ass:ivqr-A2}), and $\bP (\tilde U_{d_0} \in B \mid \tilde X) = \mu (B)$ by (\ref{ivmqr:ass:ivqr-A1}), which concludes.

Conversely, suppose that (\ref{ivmqr:eq:testable-representation}) holds for any measurable subset $B \subset \cU.$
Define $\tilde U \coloneqq q_D^{-1} (Y, X),$ where the inverse is with respect to the first variable.
By (\ref{ivmqr:eq:testable-representation}), it follows that
$$
    \bP (\tilde U \in B \mid Z, X)
    =
    \bP (Y \in q_D (B, X) \mid Z, X)
    =
    \mu (B)
    .
$$
For each $d,$ let $\tilde U_d \coloneqq \tilde U$ and $\tilde Y_d \coloneqq q_d (\tilde U_d, X).$
Then, $((\tilde Y_d)_{d \in \cD}, D, X, Z, (\tilde U_d)_{d \in \cD})$ satisfy Assumption \ref{ivmqr:ass:ivqr} with the structural function $q.$
Let $\tilde Y \coloneqq \tilde Y_D.$
It suffices to show that $\tilde Y \mid D, X, Z \overset{d}{=} Y \mid D, X, Z.$
This equivalence follows since
$$
    \cL (\tilde Y \mid D, X, Z)
    =
    \cL (q_D (\tilde U, X) \mid D, X, Z)
    =
    \cL (Y \mid D, X, Z)
$$
holds.
\end{proof}

\subsection{Proof of Theorem \ref{ivmqr:thm:local-identification-IVQR-binary}} \label{ivmqr:sec:proof-identification-IVQR}

\begin{proof}
We first show that any observationally equivalent structural function has a larger range than the true one does.
\begin{lemma} \label{ivmqr:lem:support-identification}
Let $d \in \cD.$ 
Suppose that $q$ is observationally equivalent to $q^\ast.$ 
Then,
\begin{equation}
    \bigcup_{z\in\cZ}
    \overline{\{y \in \bR^p \mid f_{d, z} (y) > 0\}}
    \subseteq q_d(\cU)
\end{equation}
holds.
Moreover, under Assumption \ref{ivmqr:ass:div-free},
\begin{align*}
    \cY_d
    =
    q_d^\ast (\cU)
    =
    \bigcup_{z \in \cZ}
    \overline{\{y \in \bR^p \mid f_{d, z} (y) > 0\}}
\end{align*}
holds, and consequently, $\cY_d \subset q_d (\cU)$ for any observationally equivalent $q.$
\end{lemma}

This lemma implies that any candidate $q_d$ must have a larger range than the true structural function does.
Thus, it is enough to show the local identification in 
\begin{align*}
    \cQ_{0, s}
    \coloneqq
    \left\{
        q^\ast + \alpha h \in \tilde \cQ_s
        \mid
        \alpha \geq 0
        ,
        \norm{h}_\cQ = 1
        ,
        \max_{d \in \cD} \norm{D h_d}_\infty \leq K
    \right\}
    ,
\end{align*}
where $\tilde \cQ_s \coloneqq \{q \in \tilde \cQ \mid \cY_d \subset q_d (\cU)\}.$

Next, we show that $q^\ast$ is the locally unique solution to the system (\ref{ivmqr:eq:system-monge-ampere}) of equations in $\cQ_{0, s}$ using Theorem \ref{ivmqr:thm:local-identification-general}.
Recall that $(\cQ, \norm{\cdot}_\cQ)$ is a normed space and that $\cQ_{0, s} \subset \cQ.$

For $z \in \cZ,$ define $\phi_z : \cQ_{0, s} \to (C^1 (\cU))^\ast$ as
\begin{align*}
    \phi_z (q) (\psi)
    \coloneqq
    \int_\cU
        \psi (u)
        \left(
            \sum_{d \in \cD}
            f_{d, z} (q_d (u))
            \det (D q_d (u))
        \right)
    d u
\end{align*}
for $\psi \in C^1 (\cU).$
Here, $(C^1 (\cU))^\ast$ is the topological dual space of $C^1 (\cU).$
It is a normed space with the norm
$$
    \norm{\nu}_{(C^1)^\ast}
    \coloneqq
    \sup_{\norm{f}_{C^1} \leq 1}
    \int_\cU f d \nu
    .
$$

The identification restriction (\ref{ivmqr:eq:system-monge-ampere}) is written as $\phi_z (q) = \mu$ for all $z \in \cZ.$
We apply Theorem \ref{ivmqr:thm:local-identification-general} to this system of equations where $\cA = \cQ,$ $\cA_0 = \cQ_{0, s},$ $\cB = (C^1 (\cU))^\ast,$ and $b_z = \mu.$
The normalized tangent set that appears in Appendix \ref{ivmqr:app:local-identification-system} is
\begin{align*}
    T_{q^\ast}^1
    \coloneqq
    \left\{
        h \in \cQ
        \mid
        \norm{h}_\cQ = 1
        ,
        \max_{d \in \cD} \norm{D h_d}_\infty \leq K
        ,
        \exists \alpha > 0
        \ \text{s.t.} \ 
        q^\ast + \alpha h \in \tilde \cQ_s
    \right\}
    .
\end{align*}

We need to check Assumptions \ref{ivmqr:ass:differentiable} and \ref{ivmqr:ass:full-rank}.
Remember that for a smooth vector field $V : \cU \to \bR^p,$ the divergence of $V$ at $u \in \cU$ is defined as
\begin{align*}
    \div V (u)
    \coloneqq
    \sum_{i = 1}^p
    \frac{\partial V_i}{\partial u_i} (u)
\end{align*}

The following lemma states that $\phi_z$ is differentiable in the sense of Assumption \ref{ivmqr:ass:differentiable}.
\begin{lemma} \label{ivmqr:lem:phi-is-differentiable}
Under Assumptions \ref{ivmqr:ass:correct-specification} and \ref{ivmqr:ass:smooth-densities}, the operator $\phi_z$ satisfies Assumption \ref{ivmqr:ass:differentiable} with
\begin{align*}
    (\phi_z)_{q^\ast}^\prime (h) (\psi)
    &\coloneqq
    -
    \int_\cU
        (\nabla \psi(u))^\prime
        \left(
            \sum_{d \in \cD}
            f_{d,z} (q_d^\ast (u))
            \cof (D q_d^\ast (u))
            h_d (u)
        \right)
    d u
    .
\end{align*}
for $h \in T_{q^\ast}^1$ and $\psi \in C^1 (\cU).$
\end{lemma}

Next, we verify Assumption \ref{ivmqr:ass:full-rank} by contradiction.
Suppose there is a sequence $h^n = (h_d^n)_{d \in \cD} \in T_{q^\ast}^1$ such that 
\begin{align} \label{ivmqr:eq:assumption-for-contradiction}
    \sum_{z \in \cZ}
    \norm{(\phi_z)^\prime_{q^\ast} (h^n)}_{(C^1)^\ast}
    \to
    0
\end{align}
as $n \to \infty.$

By the definition of $T_{q^\ast}^1,$ there is a sequence $q^n \in \tilde \cQ_s$ such that $h_d^n = \frac{q_d^n - q_d^\ast}{\norm{q^n - q^\ast}_\cQ} = \frac{\nabla (\varphi_d^n - \varphi_d^\ast)}{\norm{q^n - q^\ast}_\cQ}$ where $\nabla \varphi_d^n = q_d^n$ and $\nabla \varphi_d^\ast = q_d^\ast$ such that $\varphi_d^n (u_0) = \varphi_d^\ast (u_0) = 0$ for some fixed $u_0 \in \cU.$
Let $\kappa_d^n \coloneqq \frac{\varphi_d^n - \varphi_d^\ast}{\norm{q^n - q^\ast}_\cQ}.$
Since $\norm{\kappa_d^n}_\infty \leq 1,$ $\norm{\nabla \kappa_d^n}_\infty = \norm{h_d^n}_\infty \leq 1$ and $\norm{D^2 \kappa_d^n}_\infty = \norm{D h_d^n}_\infty \leq K,$ there is $\kappa_d^\infty \in C^1 (\cU)$ such that $\norm{\kappa_d^n - \kappa_d^\infty}_\infty + \norm{\nabla \kappa_d^n - \nabla \kappa_d^\infty}_\infty \to 0$ by taking a subsequence if necessary.

It holds that
\begin{align*}
    \norm{
        (\phi_z)_{q^\ast}^\prime (h^n)
        -
        (\phi_z)_{q^\ast}^\prime (h^\infty)
    }_{(C^1)^\ast}
    &=
    \sup_{\norm{\psi}_{C^1} \leq 1}
    \int_\cU
        (\nabla \psi (u))^\prime
        \sum_{d \in \cD}
        f_{d, z} (q_d^\ast (u))
        \cof (D q_d^\ast (u))
        (h_d^n (u) - h_d^\infty (u))
    d u
    \\
    &\leq
    C
    \norm{h_d^n - h_d^\infty}_\infty
    \\
    &\leq
    C
    \norm{h^n - h^\infty}_\infty
\end{align*}
for some positive constant $C > 0$ that depends only on the true data-generating process.
Since $h^n \to h^\infty,$ we have
$$
    (\phi_z)_{q^\ast}^\prime (h^n)
    \to
    (\phi_z)_{q^\ast}^\prime (h^\infty)
$$
in $(C^1 (\cU))^\ast.$
By condition (\ref{ivmqr:eq:assumption-for-contradiction}), we obtain $(\phi_z)_{q^\ast}^\prime (h^\infty) = 0.$
Since $h^\infty = \nabla \kappa^\infty,$ this condition implies that for any $z \in \cZ$ and $\psi \in C^1 (\cU),$
$$
    \int_\cU
        (\nabla \psi (u))^\prime
        \left(
            \sum_{d \in \cD}
            f_{d, z} (q_d^\ast (u))
            \cof (D q_d^\ast (u))
            \nabla \kappa_d^\infty (u)
        \right)
    d u
    =
    0
    .
$$
By Assumption \ref{ivmqr:ass:div-free}, we have $\nabla \kappa^\infty \equiv 0.$

However, this contradicts the fact that $h^n \in T_{q^\ast}^1,$ which particularly implies 
$$
    1 
    = 
    \max_{d \in \cD} \norm{h_d^n}_\infty
    =
    \max_{d \in \cD} \norm{h_d^n - \nabla \kappa_d^\infty}_\infty
    =
    \max_{d \in \cD} \norm{\nabla \kappa_d^n - \nabla \kappa_d^\infty}_\infty
    \to 
    0
    .
$$
Thus, Assumption \ref{ivmqr:ass:full-rank} is satisfied.

Now, Theorem \ref{ivmqr:thm:local-identification-general} implies that $q^\ast$ is the locally unique solution to (\ref{ivmqr:eq:system-monge-ampere}) in $\cQ_{0, s}.$
Theorem \ref{ivmqr:thm:local-identification-IVQR-binary} follows immediately.
\end{proof}

\subsection{Proof of Proposition \ref{ivmqr:prop:sufficient-condition-zero}}

\begin{proof}
Suppose that $\kappa_0, \kappa_1 \in C^1 (\cU)$ satisfy condition (\ref{ivmqr:eq:div-free}).
Since $\bP (D = 1 \mid Z = 0) = 0,$ it follows, by choosing $\psi = \kappa_0$ and $z = 0,$ that
$$
    \int_\cU
        (\nabla \kappa_0 (u))^\prime
        \left(
            f_{0, 0} (q_0^\ast (u))
            \cof (D q_0^\ast (u))
        \right)
        \nabla \kappa_0 (u)
    du
    =
    0
    .
$$
As the matrix $f_{0, 0} (q_0^\ast (u)) \cof (D q_0^\ast (u)) = f_{0, 0} (q_0^\ast (u)) \det (D q_0^\ast (u)) (D q_0^\ast (u))^{-1}$ is positive definite, we have $\nabla \kappa_0 \equiv 0.$
If we set $\psi = \kappa_1$ and $z = 1$ in (\ref{ivmqr:eq:div-free}), $\nabla \kappa_0 \equiv 0$ implies
$$
    \int_\cU
        (\nabla \kappa_1 (u))^\prime
        \left(
            f_{1, 1} (q_1^\ast (u))
            \cof (D q_1^\ast (u))
        \right)
        \nabla \kappa_1 (u)
    du
    =
    0
    ,
$$
which leads to $\nabla \kappa_1 \equiv 0.$
\end{proof}

\subsection{Proof of Lemma \ref{ivmqr:lem:mu-observational-equivalence}}

\begin{proof}
Suppose that $(\tilde \mu_{d, z})_{d, z}$ is observationally equivalent to $(\mu_{d, z}^\ast)_{d, z}.$
Then, there exists $((\tilde Y_d), \tilde D, \tilde Z, (\tilde U_d))$ satisfying Assumption \ref{ivmqr:ass:ivqr} with $\tilde U \mid \tilde D = d, \tilde Z = z \sim \tilde \mu_{d, z}.$
By Theorem \ref{ivmqr:thm:representation}, it follows that 
\begin{align*}
    \mu
    &=
    \cL (\tilde U \mid \tilde Z = z)
    \\
    &=
    \bP (\tilde D = 0 \mid \tilde Z = z)
    \cL (\tilde U \mid \tilde D = 0, \tilde Z = z)
    +
    \bP (\tilde D = 1 \mid \tilde Z = z)
    \cL (\tilde U \mid \tilde D = 1, \tilde Z = z)
    \\
    &=
    p_{0, z}
    \tilde \mu_{0, z}
    +
    p_{1, z}
    \tilde \mu_{1, z}
    ,
\end{align*}
which ensures the first condition.
Since $\tilde Y_d = \nabla \tilde \varphi_d (\tilde U_d)$ holds for some strictly convex function $\tilde \varphi_d$, the independence between $\tilde U_d$ and $\tilde Z$ implies
$$
    \nu_{d, z}
    =
    \cL (\tilde Y \mid \tilde D = d, \tilde Z = z)
    =
    \cL (\tilde Y_d \mid \tilde D = d, \tilde Z = z)
    =
    \cL (\nabla \tilde \varphi_d (\tilde U_d) \mid D = d, Z = z)
    =
    (\nabla \tilde \varphi_d)_\# \tilde \mu_{d, z}
$$
for all $(d, z),$ which shows the second condition.

Conversely, suppose that $(\tilde \mu_{d, z})_{d, z}$ satisfies the two conditions.
For each $d,$ by Condition 2, there exist a convex function $\tilde \varphi_d$ and a coupling $\tilde \pi_{d, 0}, \tilde \pi_{d, 1}$ such that $\tilde \pi_{d, z}$ is a coupling between $\tilde \mu_{d, z}$ and $\tilde \nu_{d, z}$ and $\supp (\tilde \pi_{d, z}) \subset \partial \tilde \varphi_d$ hold for each $z.$
On a sufficiently rich probability space, define $(\tilde D, \tilde Z) \overset{d}{=} (D, Z)$ and $\tilde U_{d, z} \sim \tilde \mu_{d, z}$ so that $(\tilde U_{d, z})_{d, z}$ are independent, and $(\tilde D, \tilde Z) \indep (\tilde U_{d, z})_{d, z}.$
Let $\tilde U \coloneqq \tilde U_{\tilde D, \tilde Z}.$
Observe that
$$
    \cL (\tilde U) 
    =
    \sum_{d, z}
    \bP (\tilde D = d, \tilde Z = z)
    \tilde \mu_{d, z}
    =
    \mu
    ,
$$
where the last equality holds by Condition 1.

Let $\tilde U_d \coloneqq \tilde U$ and $\tilde Y_d \coloneqq \nabla \tilde \varphi_d (\tilde U_d),$ which is well-defined since $\tilde U_d \sim \mu$ and $\mu$ is absolutely continuous with respect to the Lebesgue measure.
Define $\tilde Y \coloneqq \tilde Y_{\tilde D}.$
Then, by Condition 1, it follows that
\begin{align*}
    \cL (\tilde U_d \mid \tilde Z = z)
    &=
    \cL (\tilde U \mid \tilde Z = z)
    \\
    &=
    \cL (\tilde U_{\tilde D, z} \mid \tilde Z = z)
    \\
    &=
    \bP (\tilde D = 0 \mid \tilde Z = z)
    \cL (\tilde U_{0, z} \mid \tilde D = 0, \tilde Z = z)
    +
    \bP (\tilde D = 1 \mid \tilde Z = z)
    \cL (\tilde U_{1, z} \mid \tilde D = 1, \tilde Z = z)
    \\
    &=
    \bP (D = 0 \mid Z = z)
    \cL (\tilde U_{0, z})
    +
    \bP (D = 1 \mid Z = z)
    \cL (\tilde U_{1, z})
    \\
    &=
    p_{0, z}
    \tilde \mu_{0, z}
    +
    p_{1, z}
    \tilde \mu_{1, z}
    \\
    &=
    \mu
    ,
\end{align*}
which implies $\tilde U_d \mid \tilde Z \sim \mu.$

Also, since 
\begin{align*}
    \cL (\tilde Y \mid \tilde D = d, \tilde Z = z)
    &=
    \cL (\tilde Y_d \mid \tilde D = d, \tilde Z = z)
    \\
    &=
    \cL (\nabla \tilde \varphi_d (\tilde U) \mid \tilde D = d, \tilde Z = z)
    \\
    &=
    (\nabla \tilde \varphi_d)_\# 
    \cL (\tilde U \mid \tilde D = d, \tilde Z = z)
    \\
    &=
    (\nabla \tilde \varphi_d)_\# 
    \cL (\tilde U_{d, z})
    \\
    &=
    (\nabla \tilde \varphi_d)_\# 
    \tilde \mu_{d, z}
    \\
    &=
    \nu_{d, z}
    ,
\end{align*}
it follows that $(\tilde D, \tilde Z, \tilde Y) \overset{d}{=} (D, Z, Y).$ 
\end{proof}

\subsection{Proof of Lemma \ref{ivmqr:lem:L-characterization}}

\begin{proof}
Let $\pi_{d, z}$ be an optimal coupling of $\mu_{d, z}$ and $\nu_{d, z}.$
Then, $(\pi_{d, 0} + \pi_{d, 1}) / 2$ is a coupling of $(\mu_{d, 0} + \mu_{d, 1}) / 2$ and $(\nu_{d, 0} + \nu_{d, 1}) / 2.$
Hence, it holds that
\begin{align*}
    W_2^2 (\mu_{d, 0}, \nu_{d, 0})
    +
    W_2^2 (\mu_{d, 1}, \nu_{d, 1})
    &=
    \int \norm{u - y}^2 d \pi_{d, 0} (u, y)
    +
    \int \norm{u - y}^2 d \pi_{d, 1} (u, y)
    \\
    &=
    2
    \int \norm{u - y}^2 d \left(\frac{\pi_{d, 0} + \pi_{d, 1}}{2}\right) (u, y)
    \\
    &\geq
    2 W_2^2 \left(\frac{\mu_{d, 0} + \mu_{d, 1}}{2}, \frac{\nu_{d, 0} + \nu_{d, 1}}{2}\right)
    ,
\end{align*}
which is equivalent to $L_d (\mu_{d, 0}, \mu_{d, 1}) \geq 0.$

Suppose that $(\mu_{d, 0}, \mu_{d, 1})$ and $(\nu_{d, 0}, \nu_{d, 1})$ are simultaneously optimally coupleable.
Let $\pi_{d, z}$ be an optimal coupling of $\mu_{d, z}$ and $\nu_{d, z}.$
There is a convex function $\varphi_d$ such that $\supp (\pi_{d, 0}), \supp (\pi_{d, 1}) \subset \partial \varphi_d,$ and therefore, $\supp ((\pi_{d, 0} + \pi_{d, 1}) / 2) \subset \partial \varphi_d.$
Since $\pi_d \coloneqq (\pi_{d, 0} + \pi_{d, 1}) / 2$ is an optimal coupling between $(\mu_{d, 0} + \mu_{d, 1}) / 2$ and $(\nu_{d, 0} + \nu_{d, 1}) / 2,$ it follows that 
\begin{align*}
    2 W_2^2 \left(\frac{\mu_{d, 0} + \mu_{d, 1}}{2}, \frac{\nu_{d, 0} + \nu_{d, 1}}{2}\right)
    &=
    2 \int \norm{u - y}^2 d \pi_d (u, y)
    \\
    &=
    \int \norm{u - y}^2 d \pi_{d, 0} (u, y)
    +
    \int \norm{u - y}^2 d \pi_{d, 1} (u, y)
    \\
    &=
    W_2^2 (\mu_{d, 0}, \nu_{d, 0})
    +
    W_2^2 (\mu_{d, 1}, \nu_{d, 1})
    ,
\end{align*}
which is equivalent to $L_d (\mu_{d, 0}, \mu_{d, 1}) = 0.$

Conversely, suppose that $L_d(\mu_{d, 0}, \mu_{d, 1}) = 0$ holds.
Let $\varphi_d$ be a convex optimal transport potential from $(\mu_{d, 0} + \mu_{d, 1}) / 2$ to $(\nu_{d, 0} + \nu_{d, 1}) / 2.$
Note that $\varphi_d$ is feasible for the dual of the optimal transport problem between $\mu_{d, z}$ and $\nu_{d, z}$ for both $z \in \{0, 1\}.$
If it is not optimal for either problem, then $L_d(\mu_{d, 0}, \mu_{d, 1}) > 0;$ hence $\varphi_d$ is an optimal transport potential for both problems.
It follows that $\supp (\pi_{d, z}) \subset \partial \varphi_d.$
\end{proof}

\subsection{Proof of Theorem \ref{ivmqr:thm:observational-equivalence-mu}}

\begin{proof}
Suppose that $(\mu_{d, z})_{d, z}$ is observationally equivalent to $(\mu_{d, z}^\ast)_{d, z}.$
Condition 1 of Lemma \ref{ivmqr:lem:mu-observational-equivalence} implies that $(\mu_{d, z})_{d, z}$ is feasible for (\ref{ivmqr:eq:population-L-global}).
By Condition 2 of Lemma \ref{ivmqr:lem:mu-observational-equivalence} and the second part of Lemma \ref{ivmqr:lem:L-characterization}, the objective function attains zero at $(\mu_{d, z})_{d, z}.$
Since the objective is non-negative by the first part of Lemma \ref{ivmqr:lem:L-characterization}, $(\mu_{d, z})_{d, z}$ is optimal.

Conversely, suppose that $(\mu_{d, z})_{d, z}$ is an optimal solution.
It suffices to confirm the two conditions in Lemma \ref{ivmqr:lem:mu-observational-equivalence}.
The feasibility of $(\mu_{d, z})_{d, z}$ implies the first condition.
The second condition is guaranteed as follows.
Since the objective, which is non-negative by Lemma \ref{ivmqr:lem:L-characterization}, is zero at $(\mu_{d, z}^\ast)_{d, z},$ it is zero at $(\mu_{d, z})_{d, z}$ by the optimality.
By Lemma \ref{ivmqr:lem:L-characterization}, it follows that $(\mu_{d, 0}, \mu_{d, 1})$ and $(\nu_{d, 0}, \nu_{d, 1})$ are simultaneously optimally coupleable.
\end{proof}

\subsection{Proof of Theorem \ref{ivmqr:thm:q-uniform-consistency}}

\begin{proof}
We start with claim (i).
The key strategy is to apply the following lemma.

\begin{lemma} \label{ivmqr:lem:argmin-consistency}
Let $(\Theta, d)$ be a compact metric space and $\Theta_0, \Theta_n \subset \Theta$ be its subsets.
Suppose that $Q_0 : \Theta_0 \to \bR$ and $Q_n : \Theta_n \to \bR$ satisfy the following set of conditions.
\begin{enumerate}
    \item $Q_0$ has a unique minimizer $\theta_0 \in \Theta_0.$
    \item There exists a sequence $\bar \theta_n \in \Theta_n$ such that $\bar \theta_n \to \theta_0.$
    \item If $\theta_{n_k} \in \Theta_{n_k}$ converges to $\bar \theta \in \Theta,$ then $\bar \theta \in \Theta_0.$ 
    \item For $\theta_n \in \Theta_n,$ if $\theta_n \to \theta,$ then $Q_n (\theta_n) \to Q_0 (\theta).$
\end{enumerate}
If $\hat \theta_n \in \argmin_{\theta_n \in \Theta_n} Q_n (\theta_n),$ then $\hat \theta_n \to \theta_0.$
\end{lemma}

Below, we assume $p_{d, z} > 0$ for all $d, z,$ but the same argument works for other cases.

Fix a sample at which $\hat p_{d, z} \to p_{d, z}$ and $W_2 (\hat \nu_{d, z}, \nu_{d, z}) \to 0$ hold for $(d, z) \in \cD \times \cZ,$ both of which occur almost surely.
Let $\Theta \coloneqq \cP (\cU) \times \cP (\cU),$ which is compact under the Wasserstein distance $W_2 (\cdot, \cdot).$ Define
\begin{align*}
    \Theta_n
    &\coloneqq
    \left\{
        (\mu_{0, 1}, \mu_{1, 0})
        \in 
        \Theta
        \ \bigg | \
        \left(
            \frac{\mu - \hat p_{1, 0} \mu_{1, 0}}{\hat p_{0, 0}},
            \mu_{0, 1},
            \mu_{1, 0},
            \frac{\mu - \hat p_{0, 1} \mu_{0, 1}}{\hat p_{1, 1}}
        \right)
        \in
        \cM
    \right\}
    ,
    \\
    \Theta_0
    &\coloneqq
    \left\{
        (\mu_{0, 1}, \mu_{1, 0})
        \in 
        \Theta
        \ \bigg | \
        \left(
            \frac{\mu - p_{1, 0} \mu_{1, 0}}{p_{0, 0}},
            \mu_{0, 1},
            \mu_{1, 0},
            \frac{\mu - p_{0, 1} \mu_{0, 1}}{p_{1, 1}}
        \right)
        \in
        \cM
    \right\}
    ,
\end{align*}
and
\begin{align*}
    Q_n (\mu_{0, 1}, \mu_{1, 0})
    &\coloneqq
    \hat L_0 \left(
        \frac{\mu - \hat p_{1, 0} \mu_{1, 0}}{\hat p_{0, 0}},
        \mu_{0, 1}
    \right)
    +
    \hat L_1 \left(
        \mu_{1, 0},
        \frac{\mu - \hat p_{0, 1} \mu_{0, 1}}{\hat p_{1, 1}}
    \right)
    ,
    \\
    Q_0 (\mu_{0, 1}, \mu_{1, 0})
    &\coloneqq
    L_0 \left(
        \frac{\mu - p_{1, 0} \mu_{1, 0}}{p_{0, 0}},
        \mu_{0, 1}
    \right)
    +
    L_1 \left(
        \mu_{1, 0},
        \frac{\mu - p_{0, 1} \mu_{0, 1}}{p_{1, 1}}
    \right)
    .
\end{align*}
Then, $(\hat \mu_{0, 1}, \hat \mu_{1, 0}) \in \argmin_{(\mu_{0, 1}, \mu_{1, 0}) \in \Theta_n} Q_n (\mu_{0, 1}, \mu_{1, 0})$ and $(\mu_{0, 1}^\ast, \mu_{1, 0}^\ast) \in \argmin_{(\mu_{0, 1}, \mu_{1, 0}) \in \Theta_0} Q_0 (\mu_{0, 1}, \mu_{1, 0})$ under the unique solvability and $p_{0, 0}, p_{1, 1} > 0.$

To apply Lemma \ref{ivmqr:lem:argmin-consistency} to this setting, we check the four conditions.

\noindent
\textbf{Condition 1}.
$(\mu_{0, 1}^\ast, \mu_{1, 0}^\ast)$ is the unique minimizer of $Q_0$ by the definition of $\cM.$

\noindent
\textbf{Condition 2}.
Define the signed measures
$$
    \check\mu_{0,1,n}
    \coloneqq
    \frac{\mu-\hat p_{1,1}\mu_{1,1}^\ast}{\hat p_{0,1}}
    ,
    \quad
    \check\mu_{1,0,n}
    \coloneqq
    \frac{\mu-\hat p_{0,0}\mu_{0,0}^\ast}{\hat p_{1,0}}
    .
$$
If the positive part of a signed measure $\eta$ is denoted by $\eta^+$, define
$$
    \bar\mu_{0,1,n}
    \coloneqq
    \frac{\check\mu_{0,1,n}^+}
         {\check\mu_{0,1,n}^+(\mathcal U)}
    ,
    \quad
    \bar\mu_{1,0,n}
    \coloneqq
    \frac{\check\mu_{1,0,n}^+}
         {\check\mu_{1,0,n}^+(\mathcal U)}
     .
$$
It is easy to show that $\bar \mu_{0, 1, n}$ and $\bar \mu_{1, 0, n}$ are probabilities, and $(\bar \mu_{0, 1, n}, \bar \mu_{1, 0, n}) \to (\mu_{0, 1}^\ast, \mu_{1, 0}^\ast) \in \Theta_0.$
Note that the convergence holds since $\mu = p_{0, z} \mu_{0, z}^\ast + p_{1, z} \mu_{1, z}^\ast$ and $\hat p_{d,z}\to p_{d,z}.$

Now, it suffices to show $(\bar \mu_{0, 1, n}, \bar \mu_{1, 0, n}) \in \Theta_n$ for sufficiently large $n.$
Observe that by construction,
$$
    \frac{\mu-\hat p_{1,0}\bar\mu_{1,0,n}}{\hat p_{0,0}}
    \quad
    \text{and}
    \quad
    \frac{\mu-\hat p_{0,1}\bar\mu_{0,1,n}}{\hat p_{1,1}}
$$
are probability measures.
Moreover, it is clear that
$$
    \left(
        \frac{\mu-\hat p_{1,0}\bar\mu_{1,0,n}}{\hat p_{0,0}},
        \bar \mu_{0, 1, n},
        \bar \mu_{1, 0, n},
        \frac{\mu-\hat p_{0,1}\bar\mu_{0,1,n}}{\hat p_{1,1}}
    \right)
    \to
    \left(
        \mu_{0, 0}^\ast,
        \mu_{0, 1}^\ast,
        \mu_{1, 0}^\ast,
        \mu_{1, 1}^\ast
    \right)
    \in
    \cM
    .
$$
By the assumption on $\cM,$ it holds that
$$
    \left(
        \frac{\mu-\hat p_{1,0}\bar\mu_{1,0,n}}{\hat p_{0,0}},
        \bar \mu_{0, 1, n},
        \bar \mu_{1, 0, n},
        \frac{\mu-\hat p_{0,1}\bar\mu_{0,1,n}}{\hat p_{1,1}}
    \right)
    \in
    \cM
$$
for sufficiently large $n,$ which implies $(\bar \mu_{0,1,n}, \bar \mu_{1,0,n}) \in \Theta_n.$

\noindent
\textbf{Condition 3}.
Suppose that $(\mu_{0, 1}^{n_k}, \mu_{1, 0}^{n_k}) \in \Theta_{n_k}$ converges to $(\bar \mu_{0, 1}, \bar \mu_{1, 0}) \in \Theta.$
By $\cM$ being compact, it follows that
$$
    \cM
    \ni
    \left(
        \frac{\mu - \hat p_{1, 0} \mu_{1, 0}^{n_k}}{\hat p_{0, 0}},
        \mu_{0, 1}^{n_k},
        \mu_{1, 0}^{n_k},
        \frac{\mu - \hat p_{0, 1} \mu_{0, 1}^{n_k}}{\hat p_{1, 1}}
    \right)
    \to
    \left(
        \frac{\mu - p_{1, 0} \bar \mu_{1, 0}}{p_{0, 0}},
        \bar \mu_{0, 1},
        \bar \mu_{1, 0},
        \frac{\mu - p_{0, 1} \bar \mu_{0, 1}}{p_{1, 1}}
    \right)
    \in
    \cM
    .
$$
Hence, $(\bar \mu_{0, 1}, \bar \mu_{1, 0}) \in \Theta_0$ holds.

\noindent
\textbf{Condition 4}.
Suppose that $(\mu_{0, 1}^n, \mu_{1, 0}^n) \in \Theta_n$ converges to $(\mu_{0, 1}, \mu_{1, 0}) \in \Theta_0.$
Observe that 
\begin{align*}
    \left|
        Q_n (\mu_{0, 1}^n, \mu_{1, 0}^n)
        -
        Q_0 (\mu_{0, 1}, \mu_{1, 0})
    \right|
    &\leq
    \left|
        \hat L_0 \left(
            \frac{\mu - \hat p_{1, 0} \mu_{1, 0}^n}{\hat p_{0, 0}}
            ,
            \mu_{0, 1}^n
        \right)
        -
        L_0 \left(
            \frac{\mu - p_{1, 0} \mu_{1, 0}}{p_{0, 0}}
            ,
            \mu_{0, 1}
        \right)
    \right|\\
    &\hspace{2cm}
    +
    \left|
        \hat L_1 \left(
            \mu_{1, 0}^n
            ,
            \frac{\mu - \hat p_{0, 1} \mu_{0, 1}^n}{\hat p_{1, 1}}
        \right)
        -
        L_1 \left(
            \mu_{1, 0}
            ,
            \frac{\mu - p_{0, 1} \mu_{0, 1}}{p_{1, 1}}
        \right)
    \right|
    ,
\end{align*}
recalling the definition
\begin{align*}
    \hat L_d (\mu_{d, 0}, \mu_{d, 1})
    &\coloneqq
    W_2^2 (\mu_{d, 0}, \hat \nu_{d, 0})
    +
    W_2^2 (\mu_{d, 1}, \hat \nu_{d, 1})
    -
    2
    W_2^2 \left(\frac{\mu_{d, 0} + \mu_{d, 1}}{2}, \frac{\hat \nu_{d, 0} + \hat \nu_{d, 1}}{2}\right)
    ,
    \\
    L_d (\mu_{d, 0}, \mu_{d, 1})
    &\coloneqq
    W_2^2 (\mu_{d, 0}, \nu_{d, 0})
    +
    W_2^2 (\mu_{d, 1}, \nu_{d, 1})
    -
    2 
    W_2^2 \left(\frac{\mu_{d, 0} + \mu_{d, 1}}{2}, \frac{\nu_{d, 0} + \nu_{d, 1}}{2}\right)
    .
\end{align*}
We focus on showing that the first term vanishes as $n \to \infty,$ as the same argument works for the second term.
We decompose the first term as follows:
$$
    \left|
        \hat L_0 \left(
            \frac{\mu - \hat p_{1, 0} \mu_{1, 0}^n}{\hat p_{0, 0}}
            ,
            \mu_{0, 1}^n
        \right)
        -
        L_0 \left(
            \frac{\mu - p_{1, 0} \mu_{1, 0}}{p_{0, 0}}
            ,
            \mu_{0, 1}
        \right)
    \right|
    \leq
    D_1 + D_2
$$
where
\begin{align*}
    D_1
    &\coloneqq
    \left|
        \hat L_0 \left(
            \frac{\mu - \hat p_{1, 0} \mu_{1, 0}^n}{\hat p_{0, 0}}
            ,
            \mu_{0, 1}^n
        \right)
        -
        L_0 \left(
            \frac{\mu - \hat p_{1, 0} \mu_{1, 0}^n}{\hat p_{0, 0}}
            ,
            \mu_{0, 1}^n
        \right)
    \right|
    ,
    \\
    D_2
    &\coloneqq
    \left|
        L_0 \left(
            \frac{\mu - \hat p_{1, 0} \mu_{1, 0}^n}{\hat p_{0, 0}}
            ,
            \mu_{0, 1}^n
        \right)
        -
        L_0 \left(
            \frac{\mu - p_{1, 0} \mu_{1, 0}}{p_{0, 0}}
            ,
            \mu_{0, 1}
        \right)
    \right|
    .
\end{align*}
By the continuity of $L_0,$ we have $D_2 \to 0.$
The triangle inequality implies
\begin{align*}
    D_1 
    &\lesssim
    \left|
        W_2 \left(
            \frac{1}{2}
            \left(
                \frac{\mu - \hat p_{1, 0} \mu_{1, 0}^n}{\hat p_{0, 0}}
                +
                \mu_{0, 1}^n
            \right)
            , 
            \frac{\hat \nu_{0, 0} + \hat \nu_{0, 1}}{2}
        \right)
        -
        W_2 \left(
            \frac{1}{2}
            \left(
                \frac{\mu - \hat p_{1, 0} \mu_{1, 0}^n}{\hat p_{0, 0}}
                +
                \mu_{0, 1}^n
            \right)
            , 
            \frac{\nu_{0, 0} + \nu_{0, 1}}{2}
        \right)
    \right|
    \\
    &\hspace{5cm}
    +
    \left|
        W_2 \left(
            \frac{\mu - \hat p_{1, 0} \mu_{1, 0}^n}{\hat p_{0, 0}}
            , 
            \hat \nu_{0, 0}
        \right)
        -
        W_2 \left(
            \frac{\mu - \hat p_{1, 0} \mu_{1, 0}^n}{\hat p_{0, 0}}
            ,
            \nu_{0, 0}
        \right)
    \right|
    \\
    &\hspace{7cm}
    +
    \left|
        W_2 \left(
            \mu_{0, 1}^n
            , 
            \hat \nu_{0, 1}
        \right)
        -
        W_2 \left(
            \mu_{0, 1}^n
            , 
            \nu_{0, 1}
        \right)
    \right|
    \\
    &\leq
    W_2 \left( 
        \frac{\hat \nu_{0, 0} + \hat \nu_{0, 1}}{2}
        ,
        \frac{\nu_{0, 0} + \nu_{0, 1}}{2}
    \right)
    +
    W_2 \left( 
        \hat \nu_{0, 0}
        ,
        \nu_{0, 0}
    \right)
    +
    W_2 \left( 
        \hat \nu_{0, 1}
        ,
        \nu_{0, 1}
    \right)
    ,
\end{align*}
where the first inequality suppresses a universal constant that depends only on $\cU$ and $\cY_0.$
Since the last expression vanishes, it follows that $D_1 \to 0.$
Therefore, we obtain $Q_n (\mu_{0, 1}^n, \mu_{1, 0}^n) \to Q_0 (\mu_{0, 1}, \mu_{1, 0}).$

Now, by Lemma \ref{ivmqr:lem:argmin-consistency}, $(\hat \mu_{0, 1}, \hat \mu_{1, 0}) \to (\mu_{0, 1}^\ast, \mu_{1, 0}^\ast)$ holds.
Since $\hat \mu_{0, 0} = (\mu - \hat p_{1, 0} \hat \mu_{1, 0}) / \hat p_{0, 0}$ and $\hat \mu_{1, 1} = (\mu - \hat p_{0, 1} \hat \mu_{0, 1}) / \hat p_{1, 1}$ for large $n,$ it also follows that $(\hat \mu_{0, 0}, \hat \mu_{1, 1}) \to (\mu_{0, 0}^\ast, \mu_{1, 1}^\ast).$

Since the convergence holds for almost all samples, $(\hat \mu_{d, z})_{d, z}$ is strongly consistent.

\bigskip

Next, we show claim (ii).

Fix $d \in \{0, 1\}.$
Let $\hat q_d$ be the proposed estimator of $q_d^\ast,$ and $K \subset \Int (\supp (\cL (U \mid D = d)))$ be a compact set.
By Urysohn's lemma, the compactness of $K$ implies that there exist open balls $B_1, \dots, B_N, B_1^\prime, \dots, B_N^\prime$ such that $\Cl (B_i) \subset B_i^\prime$ and
$$
    K 
    \subset 
    \cup_{i = 1}^N \Cl (B_i)
    \subset 
    \cup_{i = 1}^N B_i^\prime
    \subset 
    \Int (\supp (\cL (U \mid D = d)))
    .
$$
Therefore, it suffices to show that $\hat q_d$ is consistent uniformly on $\Cl (B_i)$ for each $i.$

Recall that $\hat \varphi_d$ is an optimal transport potential
$$
    \text{from}
    \quad
    \begin{cases}
        (\hat \mu_{d, 0} + \hat \mu_{d, 1}) / 2 & (\hat p_{d, 1 - d} \neq 0) \\
        \hat \mu_{d, d} & (\hat p_{d, 1 - d} = 0)
    \end{cases}
    \quad
    \text{to}
    \quad
    \begin{cases}
        (\hat \nu _{d, 0} + \hat \nu_{d, 1}) / 2 & (\hat p_{d, 1 - d} \neq 0) \\
        \hat \nu_{d, d} & (\hat p_{d, 1 - d} = 0)
    \end{cases}
$$
and similarly, $\varphi_d^\ast$ is an optimal transport potential
$$
    \text{from}
    \quad
    \begin{cases}
        (\mu_{d, 0}^\ast + \mu_{d, 1}^\ast) / 2 & (p_{d, 1 - d} \neq 0) \\
        \mu_{d, d}^\ast & (p_{d, 1 - d} = 0)
    \end{cases}
    \quad
    \text{to}
    \quad
    \begin{cases}
        (\nu_{d, 0} + \nu_{d, 1}) / 2 & (p_{d, 1 - d} \neq 0) \\
        \nu_{d, d} & (p_{d, 1 - d} = 0)
    \end{cases}
    .
$$
By part (i), the estimators of both the source and target distributions converge to the truth in the $2$-Wasserstein distance almost surely.
Since $\mu = p_{0, z} \mu_{0, z}^\ast + p_{1, z} \mu_{1, z}^\ast,$ each distribution $\mu_{d, z}^\ast$ admits a density if $p_{d, z} > 0.$
Notice also that the support of $U \mid D = d$ is equal to that of the source distribution.
As $\varphi_d^\ast$ is differentiable, by Proposition \ref{ivmqr:prop:consistency-ot-map},
$$
    \sup_{u \in \Cl (B_i)}
    \norm{\hat q_d (u) - q_d^\ast (u)}
    \to
    0
$$
holds almost surely.
\end{proof}

\subsection{Proof of Theorem \ref{ivmqr:thm:ot-map-is-homeo}}

\begin{proof}
We break down the proof into four steps.

\noindent
\textbf{Step 1}.
We first show that $q$ is injective on $\Int (\cU).$
Let $u^1, u^2 \in \Int (\cU)$ be two distinct points.
By the strict cyclical monotonicity, we have $(u^2 - u^1)^\prime (q (u^2) - q (u^1)) > 0,$ which implies $q (u^2) \neq q (u^1).$

\noindent
\textbf{Step 2}.
Next, we show that for $u \in \Int (\cU),$ it holds that $q (u) \in \Int (\cY).$
The proof is by contradiction.
Suppose there is $u \in \Int (\cU)$ such that $q (u) \in \partial \cY.$
Since $\Int (\cY)$ is convex, there exists $c \in \bR^p$ such that $(y - q (u))^\prime c < 0$ for $y \in \Int (\cY)$ by the separating hyperplane theorem.
It follows that $v \coloneqq u + \varepsilon c \in \Int (\cU)$ for sufficiently small $\varepsilon > 0.$
By the strict cyclical monotonicity, we have $(q (v) - q (u))^\prime c > 0.$
Take a sequence $y_n \in \Int (\cY)$ such that $y_n \to q (v).$
Then, $0 \geq \lim_{n \to \infty} (y_n - q (u))^\prime c = (q (v) - q (u))^\prime c > 0,$ a contradiction.

\noindent
\textbf{Step 3}.
Now, we show that $q |_{\Int (\cU)} : \Int (\cU) \to \Int (\cY)$ is surjective.
Let $y \in \Int (\cY).$ Since $y \in \cY = q(\cU),$ there exists a unique $u \in \cU$ such that $q(u) = y.$
We must show $u \notin \partial \cU.$
Suppose for contradiction that $u \in \partial \cU.$
Since $\cU$ is closed and convex, the supporting hyperplane theorem guarantees the existence of a non-zero vector $c \in \bR^p$ such that $(v - u)^\prime c \leq 0$ for all $v \in \cU.$
Because $y = q(u)$ is in $\Int (\cY),$ there exists $\varepsilon > 0$ such that $y + \varepsilon c \in \Int(\cY) \subseteq \cY.$
Thus, there exists $v \in \cU$ such that $q(v) = y + \varepsilon c.$
By strict cyclical monotonicity, we have $(v - u)^\prime (q(v) - q(u)) > 0.$
However, substituting $q(v) - q(u) = \varepsilon c,$ we get $\varepsilon (v - u)^\prime c > 0,$ meaning $(v - u)^\prime c > 0,$ which contradicts the supporting hyperplane inequality.
Hence, $u \in \Int (\cU).$

\noindent
\textbf{Step 4}.
By Steps 1, 2, and 3, $q |_{\Int (\cU)} : \Int (\cU) \to \Int (\cY)$ is a continuous bijection.
Brouwer's invariance of domain theorem dictates that $q |_{\Int (\cU)}$ is a homeomorphism.
\end{proof}

\subsection{Proof of Proposition \ref{ivmqr:prop:consistency-ot-map}}

\begin{proof}
By Theorem \ref{ivmqr:thm:del-barrio}, there exists a sequence $a_n \in \bR$ such that $\varphi_n (x) + a_n \to \varphi (x)$ for all $x \in D.$
Since $\nabla \varphi = T$ on $D,$ we have
$$
    \sup_{x \in K}
    \norm{T_n (x) - T (x)}
    \leq
    \sup_{x \in K}
    \sup_{\xi \in \partial \varphi_n (x)}
    \norm{\xi - \nabla \varphi (x)}
    \to
    0
$$
by Lemma 3.10 of \cite{seijo2011nonparametric}.
\end{proof}

\subsection{Proof of Lemma \ref{ivmqr:lem:full-rank-version}}
\begin{proof}
The ``if'' part holds because for any $h \in T_{a^\ast}^1,$
\begin{align*}
    \sum_{z \in \cZ}
    \norm{(\phi_z)_{a^\ast}^\prime (h)}_\cB
    \geq
    \left(
        \sum_{z \in \cZ \setminus \{z_0\}}
        \norm{
            (\phi_z)_{a^\ast}^\prime (h)
        }_\cB
    \right)
    \vee
    \norm{(\phi_{z_0})_{a^\ast}^\prime (h)}_\cB 
    \geq
    \eta
    >
    0
    .
\end{align*}
To show the ``only if'' part, suppose
\begin{align*}
    \sum_{z \in \cZ \setminus \{z_0\}}
    \norm{
        (\phi_z)_{a^\ast}^\prime (h)
    }_\cB 
    < 
    \eta
    \coloneqq
    \frac{1}{2}
    \inf_{h \in T_{a^\ast}^1}
    \sum_{z \in \cZ}
    \norm{(\phi_z)_{a^\ast}^\prime (h)}_\cB
    ,
\end{align*}
where $\eta > 0$ by the hypothesis.
Then,
\begin{align*}
    \norm{(\phi_{z_0})_{a^\ast}^\prime (h)}_\cB
    &=
    \sum_{z \in \cZ}
    \norm{(\phi_z)_{a^\ast}^\prime (h)}_\cB
    -
    \sum_{z \in \cZ \setminus \{z_0\}}
    \norm{
        (\phi_z)_{a^\ast}^\prime (h)
    }_\cB 
    \\
    &>
    \inf_{h \in T_{a^\ast}^1}
    \sum_{z \in \cZ}
    \norm{(\phi_z)_{a^\ast}^\prime (h)}_\cB
    -
    \eta
    \\
    &=
    \eta
\end{align*}
holds.
\end{proof}

\subsection{Proof of Lemma \ref{ivmqr:lem:support-identification}}

\begin{proof}
Fix $d \in \cD$ and $z \in \cZ.$
Let $y \in \overline{\{y \in \bR^p \mid f_{d,z}(y) > 0\}}.$
Then, for any $\varepsilon > 0,$ it follows that 
$$
    \bP (Y \in B_y (\varepsilon), D = d, Z = z)
    >
    0
$$
where $B_y (\varepsilon)$ is the ball around $y$ with radius $\varepsilon.$
If $((\tilde Y_d)_{d \in \cD}, \tilde D, \tilde Z, (\tilde U_d)_{d \in \cD})$ is an observationally equivalent representation corresponding to $q,$ we have
$$
    0
    <
    \bP (\tilde Y \in B_y (\varepsilon), \tilde D = d, \tilde Z = z)
    \leq
    \bP (\tilde Y_d \in B_y (\varepsilon))
    \leq
    \bP (q_d (\tilde U_d) \in B_y (\varepsilon))
    ,
$$
which implies $y \in \supp (q_d (\tilde U_d)).$
It follows $y \in q_d (\cU)$ because the support of $\tilde U_d$ is equal to $\cU,$ $\cU$ is compact, and $q_d$ is continuous.
Hence, we obtain
$$
    \bigcup_{z \in \cZ}
    \overline{\{y \in \bR^p \mid f_{d,z}(y) > 0\}}
    \subset
    q_d (\cU)
    .
$$

Applying this to \(q^*\) gives
$$ 
    \bigcup_{z \in \cZ}
    \overline{\{y \in \bR^p \mid f_{d,z}(y) > 0\}}
    \subset
    q_d^\ast (\cU)
    =
    \cY_d
    ,
$$
where the last equality follows from $(q_d^\ast)_\# \mu = \cL (Y_d),$ the continuity of $q_d^\ast,$ and $\supp (\mu) = \cU.$
Suppose that the other inclusion does not hold.
Then, there exist $y \in q_d^\ast (\cU)$ and $\varepsilon > 0$ such that for every $z \in \cZ,$
$$
    \bP (Y \in B_y (\varepsilon), D = d, Z = z)
    =
    0
    .
$$
In particular, $f_{d, z} (q_d^\ast (\cdot))$ is zero on an open set $(q_d^\ast)^{-1} (B_y (\varepsilon)).$
If we choose $\kappa_d$ so that
$$
    \nabla \kappa_d (u)
    \begin{cases}
        = 0 & (u \notin (q_d^\ast)^{-1} (B_y (\varepsilon))) \\
        \not \equiv 0 & (u \in (q_d^\ast)^{-1} (B_y (\varepsilon)))
    \end{cases}
    ,
$$
condition (\ref{ivmqr:eq:div-free}) does hold, and consequently, Assumption \ref{ivmqr:ass:div-free} is violated.
Therefore, we have 
$$
    \bigcup_{z \in \cZ}
    \overline{\{y \in \bR^p \mid f_{d,z}(y) > 0\}}
    =
    \cY_d
    ,
$$
which completes the proof.
\end{proof}

\subsection{Proof of Lemma \ref{ivmqr:lem:phi-is-differentiable}}
We first show an auxiliary lemma on the uniform differentiability.
\begin{lemma} \label{ivmqr:lem:uniform-differentiability}
Let $V$ be a convex subset of an inner product space $(X, \ip{\cdot}{\cdot}).$
Suppose the function $f : V \to \bR$ is (Fr\'echet) differentiable on the interior of $V$ with a uniformly continuous derivative $D f.$ 
Then, $f$ is uniformly differentiable, i.e.,
\begin{align*}
    \lim_{\norm{h} \to 0}
    \sup_{\substack{
        x \in \Int (V)
        \\
        x + h \in V
    }}
    \left|
        \frac{f (x + h) - f (x)}{\norm{h}}
        -
        \ip{D f (x)}{\frac{h}{\norm{h}}}
    \right|
    =
    0
    .
\end{align*}
\end{lemma}

\begin{proof}
Let $ \varepsilon > 0.$
Since $D f$ is uniformly continuous, there is $\delta > 0$ such that
\begin{align*}
    \norm{x - y}
    < 
    \delta
    \Rightarrow
    \norm{D f (x) - D f (y)} 
    < 
    \varepsilon
    .
\end{align*}
Fix $x \in \Int (V),$ and choose $h$ such that $0 < |h| < \delta$ and $x + h \in V.$
Since
\begin{align*}
    f (x + h) - f (x)
    =
    \int_0^1
        \frac{d}{d t} f (x + t h)
    d t
    =
    \int_0^1
        \ip{D f (x + t h)}{h}
    d t
    ,
\end{align*}
we have
\begin{align*}
    \left|
        \frac{f (x + h) - f (x)}{\norm{h}}
        -
        \ip{D f (x)}{\frac{h}{\norm{h}}}
    \right|
    &=
    \left|
        \ip{
            \int_0^1
                D f (x + t h)
                -
                D f (x)
            d t
        }{
            \frac{h}{\norm{h}}
        }
    \right|
    \\
    &\leq
    \max_{t \in [0, 1]}
    \norm{D f (x + t h) - D f (x)}
    \\
    &<
    \varepsilon
    ,
\end{align*}
which shows the uniform differentiability.
\end{proof}

\begin{remark}
In general, the integral of $D f$ should be understood in the sense of Bochner, but this complication is not important for the proof of Lemma \ref{ivmqr:lem:phi-is-differentiable}.
\end{remark}

For matrices $A, B \in \bR^{s \times t},$ let 
\begin{align*}
    \ip{A}{B} 
    \coloneqq
    \sum_{i = 1}^s
    \sum_{j = 1}^t
    A_{i, j}
    B_{i, j}
    .
\end{align*}
For an invertible matrix $C,$ define its cofactor matrix as
\begin{align*}
    \cof (C)
    \coloneqq
    \det (C)
    (C^{-1})^\prime
    .
\end{align*}

\begin{lemma}[Transport formula]
\label{ivmqr:lem:transport-boundary}
Let $K\subset\bR^p$ be compact and convex with nonempty interior and
piecewise $C^1$ boundary.\footnote{Piecewise $C^1$ regularity is assumed to simplify the
boundary arguments. Since $\operatorname{Int}(K)$ is a bounded
Lipschitz domain (see, for example, Corollary 1.2.2.3 of \cite{grisvard2011elliptic}), the conclusion
also holds without this assumption, using the divergence theorem and
surface change-of-variables formula for Lipschitz domains.}
Let $T:K\to K$ and $b:K\to\bR$ admit $C^1$ extensions to neighborhoods
of $K$, and suppose $\norm{DT-I}_\infty\leq1/2$.
Set $S_t=(1-t)\operatorname{id}+tT$ for $t\in[0,1]$. Then
\begin{equation} \label{ivmqr:eq:transport-exact}
    \int_{T(K)}
        b (u)
    du
    -
    \int_K
        b (u)
    du
    =
    \int_0^1
        \int_{\partial K}
            b(S_t(u))(T(u)-u)^\prime
            \cof(DS_t(u))\nu_K(u)
        dS(u)
    dt
    .
\end{equation}
Moreover, for a constant $C_{K, p}$ depending only on $K$ and $p$,
\begin{equation} \label{ivmqr:eq:transport-remainder}
    \left|
        \int_{K\setminus T(K)}
            b (u) 
        du
        +
        \int_{\partial K}
            b(u) 
            (T(u)-u)^\prime
            \nu_K (u) 
        d S (u)
    \right|
    \leq 
    C_{K, p}
    \norm{b}_{C^1(K)}
    \norm{T - \id}_{C^1(K)}^2
    .
\end{equation}
Here $\nu_K$ is the outward unit normal and boundary integrals are
taken over the regular faces. 
\end{lemma}

\begin{proof}
By compactness of $K$ and continuity of $D T,$ there is a convex open neighborhood $N$ of $K$ such that
$$
    \sup_{u \in N}
    \norm{D T (u) - I}
    <
    \frac{3}{4}
    .
$$
For $t \in [0, 1]$ and $u \in N,$ let
$$
    S_t (u)
    \coloneqq
    (1-t) u 
    + 
    t T (u)
    .
$$
Since $N$ is convex, for any $u, v \in N,$  we have
\begin{align*}
    S_t (u) - S_t (v)
    &=
    (1 - t)
    (u - v)
    +
    t
    \int_0^1
        D T (s u + (1 - s) v)
    d s
    (u - v)
    \\
    &=
    u - v
    +
    t
    \int_0^1
        \left(
            D T (s u + (1 - s) v)
            -
            I
        \right)
    d s
    (u - v)
    .
\end{align*}
It follows that
\begin{align*}
    \norm{S_t (u) - S_t (v)}
    &\geq
    \norm{u - v}
    -
    \norm{S_t (u) - S_t (v) - (u - v)}
    \\
    &\geq
    \left(
        1 
        -
        t
        \sup_{u \in N} \norm{DT (u) - I}
    \right)
    \norm{u - v}
    \\
    &\geq
    \frac{1}{4}
    \norm{u - v}
    .
\end{align*}
Thus, $S_t$ is injective on $N$ for every $t \in [0, 1].$
It also holds that for every $u \in N,$ $h \in \bR^p,$ and $t \in [0, 1],$
$$
    \norm{D S_t (u) h}
    =
    \norm{h + t (D T (u) - I) h}
    \geq
    \norm{h}
    -
    t
    \norm{D T (u) - I}
    \norm{h}
    \geq
    \frac{1}{4}
    \norm{h}
    ,
$$
so $D S_t (u)$ is nonsingular.
By the inverse function theorem and its injectivity, $S_t,$ defined on $N,$ is a $C^1$ diffeomorphism onto its image for every $t \in [0, 1].$

Since $S_0 = \id$ and $S_1 = T,$ it follows that
\begin{equation} \label{ivmqr:eq:fundamental-thm-calculus}
    \int_{T(K)}
        b 
    du
    -
    \int_K
        b
    du
    =
    \int_{S_1 (K)}
        b 
    du
    -
    \int_{S_0 (K)}
        b
    du
    =
    \int_0^1
        \left(
            \frac{d}{d t}
            \int_{S_t (K)}
                b
            d u
        \right)
    d t
    .
\end{equation}
Since 
$$
    \int_{S_t (K)}
        b (u)
    d u
    =
    \int_K
        (b \circ S_t) (v)
        \det (D S_t (v))
    d v
$$
by the change-of-variable formula and $\det (D S_t) > 0,$ we have
\begin{align*}
    \frac{d}{d t}
    \int_{S_t (K)}
        b (u)
    d u
    &=
    \int_K
        \left[
            D b (S_t (v))
            D_t S_t (v)
            +
            b (S_t (v))
            \tr \left(
                (D S_t (v))^{-1}
                D_t D S_t (v)
            \right)
        \right]
        \det (D S_t (v))
    d v
    .
\end{align*}
By considering $u = S_t (v),$ the first term is evaluated as follows:
\begin{align*}
    \int_K
        D b (S_t (v))
        D_t S_t (v)
        \det (D S_t (v))
    d v
    &=
    \int_{S_t (K)}
        D b (u)
        (D_t S_t) (S_t^{-1} (u))
    d u
    \\
    &=
    \int_{S_t (K)}
        D b (u)
        ((T - \id) \circ S_t^{-1}) (u)
    d u
    \\
    &=
    \int_{S_t (K)}
        D b (u)
        V_t (u)
    d u
\end{align*}
where $V_t \coloneqq (T - \id) \circ S_t^{-1}.$
For the second term, since
\begin{align*}
    \tr \left(
        (D S_t (v))^{-1}
        D_t D S_t (v)
    \right)
    &=
    \tr \left(
        (D S_t (S_t^{-1} (u)))^{-1}
        (D_t D S_t) (S_t^{-1} (u))
    \right)
    \\
    &=
    \tr \left(
        D S_t^{-1} (u)
        (D T (S_t^{-1} (u)) - I)
    \right)
    \\
    &=
    \tr (D V_t (u))
    ,
\end{align*}
it follows that
\begin{align*}
    \int_K
        b (S_t (v))
        \tr \left(
            (D S_t (v))^{-1}
            D_t D S_t (v)
        \right)
        \det (D S_t (v))
    d v
    &=
    \int_{S_t (K)}
        b (u)
        \tr (D V_t (u))
    d u
    \\
    &=
    \int_{S_t (K)}
        b (u)
        \div V_t (u)
    d u
    .
\end{align*}
Thus, we obtain
\begin{align*}
    \frac{d}{d t}
    \int_{S_t (K)}
        b (u)
    d u
    &=
    \int_{S_t (K)}
        D b (u)
        V_t (u)
        +
        b (u)
        \div V_t (u)
    d u
    \\
    &=
    \int_{S_t (K)}
        \div (b V_t) (u)
    d u
    .
\end{align*}
By the divergence theorem, it follows
\begin{equation} \label{ivmqr:eq:divergence-thm}
    \frac{d}{d t}
    \int_{S_t (K)}
        b (u)
    d u
    =
    \int_{\partial S_t (K)}
        b (u)
        V_t (u)^\prime
        \nu_t (u)
    d S (u)
\end{equation}
where $\nu_t$ is the outward unit normal to $\partial S_t (K).$

By equation (4.1b) of \cite{berggren2023shape}, we then have
$$
    \int_{\partial S_t (K)}
        b (u)
        V_t (u)^\prime
        \nu_t (u)
    d S (u)
    =
    \int_{\partial K}
        b (S_t (v))
        V_t (S_t (v))^\prime
        \nu_t (S_t (v))
        \norm{((D S_t (v))^{\prime})^{-1} \nu_K}
        \det (D S_t (v))
    d S (v)
    .
$$
If $\rho$ is a local defining function for a regular face of $K,$ then $\rho_t \coloneqq \rho \circ S_t^{-1}$ is a local defining function for the corresponding face of $S_t(K).$
An outward normal vector at $S_t (v)$ is
$$
    D \rho_t (S_t (v))
    =
    D \rho (v)
    D S_t^{-1} (S_t (v))
    =
    \norm{D \rho (v)}
    (\nu_K (v))^\prime
    D S_t^{-1} (S_t (v))
    ,
$$
so
$$
    \nu_t (S_t (v))
    =
    \frac{(D S_t^{-1} (S_t (v)))^\prime \nu_K (v)}{\norm{(D S_t^{-1} (S_t (v)))^\prime \nu_K (v)}}
    .
$$
Since $D S_t^{-1} (S_t (v)) = (D S_t (v))^{-1},$ it follows that
\begin{align*}
    \int_{\partial S_t (K)}
        b (u)
        V_t (u)^\prime
        \nu_t (u)
    d S (u)
    &=
    \int_{\partial K}
        b (S_t (v))
        V_t (S_t (v))^\prime
        \det (D S_t (v))
        ((D S_t (v))^\prime)^{-1}
        \nu_K (v)
    d S (v)
    \\
    &=
    \int_{\partial K}
        b (S_t (v))
        V_t (S_t (v))^\prime
        \cof (D S_t (v))
        \nu_K (v)
    d S (v)
    \\
    &=
    \int_{\partial K}
        b (S_t (v))
        (T (v) - v)^\prime
        \cof (D S_t (v))
        \nu_K (v)
    d S (v)
\end{align*}
where the last equality holds by $V_t (S_t (v)) = T (v) - v.$
Combining this with equations (\ref{ivmqr:eq:fundamental-thm-calculus}) and (\ref{ivmqr:eq:divergence-thm}) yields equation (\ref{ivmqr:eq:transport-exact}).

It remains to establish (\ref{ivmqr:eq:transport-remainder}).
Since $K$ is convex and $T(K) \subseteq K,$ we have $S_t(K) \subseteq K.$
Therefore,
$$
    |b(S_t(u)) - b(u)|
    \leq
    \norm{D b}_\infty
    \norm{T - \id}_\infty
    .
$$
Also, since $D S_t = I + t (D T-I)$ and $\norm{D T-I}_\infty \leq 1/2,$ it follows that
$$
    \norm{
        \cof(D S_t) - I
    }_\infty
    \leq
    C_p
    \norm{D T-I}_\infty
$$
for a constant $C_p$ depending only on $p.$
It follows from (\ref{ivmqr:eq:transport-exact}) that
\begin{align*}
    \left|
        \int_{K \setminus T(K)}
            b (u)
        d u
        +
        \int_{\partial K}
            b(u)
            (T(u)-u)^\prime
            \nu_K(u)
        d S(u)
    \right|
    &\leq
    C_{K, p}
    \norm{b}_{C^1(K)}
    \norm{T-\id}_\infty
    \left(
        \norm{T-\id}_\infty
        +
        \norm{D T-I}_\infty
    \right)
    \\
    &\leq
    C_{K, p}
    \norm{b}_{C^1(K)}
    \norm{T-\id}_{C^1(K)}^2
    ,
\end{align*}
which proves (\ref{ivmqr:eq:transport-remainder}).
\end{proof}

\begin{proof}[Proof of Lemma \ref{ivmqr:lem:phi-is-differentiable}]
Fix $z \in \cZ$, and let $q^{\varepsilon, h} = q^\ast + \varepsilon h \in \cQ_{0, s},$ where $h \in T_{q^\ast}^1$. Extend $f_{d,z}$ by zero outside $\cY_d,$ and set
\begin{align*}
    A_{d, \varepsilon, h}
    \coloneqq
    \{u \in \cU \mid q_d^{\varepsilon, h} (u)\in\cY_d\}
    ,
    \quad
    F_{d,z,h}(u)
    \coloneqq
    f_{d,z}(q_d^\ast(u))\cof(Dq_d^\ast(u))h_d(u).
\end{align*}
The assumptions imply $\norm{\div F_{d,z,h}}_\infty\leq C$
uniformly in $h$. Splitting the density integral over
$A_{d,\varepsilon}$ and its complement and applying the divergence
theorem to the proposed derivative gives
\begin{align*}
    &
    \norm{
        \frac{\phi_z(q^\ast+\varepsilon h)-\phi_z(q^\ast)}{\varepsilon}
        -
        (\phi_z)_{q^\ast}^\prime(h)
    }_{(C^1)^\ast}
    \\
    &=
    \sup_{\norm{\psi}_{C^1} \leq 1}
    \frac{1}{\varepsilon}
    \left(
        \int_\cU
            \psi (u)
            \sum_{d \in \cD}
            (
                f_{d,z}(q_d^{\varepsilon, h}) \det(D q_d^{\varepsilon, h})
                -
                f_{d,z}(q_d^\ast) \det(Dq_d^\ast)
            )
        d u
    \right)
    +
    \int_\cU
        (\nabla \psi(u))^\prime
        \left(
            \sum_{d \in \cD}
            F_{d, z, h} (u)
        \right)
    d u
    \\
    &=
    \sup_{\norm{\psi}_{C^1} \leq 1}
    \sum_{d \in \cD}
    \left(
        \int_\cU
            \psi
            \left[
                \frac{   
                    f_{d,z}(q_d^{\varepsilon, h}) \det(Dq_d^{\varepsilon, h})
                    -
                    f_{d,z}(q_d^\ast) \det(D q_d^\ast)
                }{\varepsilon}
                -
                \div F_{d, z, h}
            \right]
        d u
        +
        \int_{\partial \cU}
            \psi
            F_{d, z, h}
            \nu
        d S
    \right)
    \\
    &\leq
    I_1
    +
    I_2
    +
    I_3
\end{align*}
where
\begin{align*}
    I_1 (\varepsilon, h)
    &\coloneqq
    \sum_{d\in\cD}\sup_{\norm{\psi}_{C^1}\leq1}
    \left|
        \int_{A_{d, \varepsilon, h}}
            \psi
            \left[
                \frac{
                 f_{d,z}(q_d^{\varepsilon, h})
                 \det(D q_d^{\varepsilon, h})
                 -
                 f_{d,z}(q_d^\ast)
                 \det(Dq_d^\ast)
                }{\varepsilon}
                -
                \div F_{d,z,h}
            \right]
        du
    \right|
    \\
    I_2 (\varepsilon, h)
    &\coloneqq
    \sum_{d\in\cD}\sup_{\norm{\psi}_{C^1}\leq1}
    \left|
        \frac{1}{\varepsilon}
        \int_{\cU\setminus A_{d, \varepsilon, h}}
            \psi f_{d,z}(q_d^\ast)\det(Dq_d^\ast)
         du
        -
        \int_{\partial\cU}
            \psi F_{d,z,h}^\prime\nu
        dS
    \right|
    \\
    I_3 (\varepsilon, h)
    &\coloneqq
    C
    \sum_{d\in\cD}
    |\cU\setminus A_{d, \varepsilon, h}|
    .
\end{align*}
Here $|\cdot|$ denotes Lebesgue volume. We bound these terms uniformly over $h\in T_{q^\ast}^1$ satisfying $q^\ast+\varepsilon h\in\cQ_{0,s}$.

For $u\in A_{d,\varepsilon}$, both $q_d^\ast(u)$ and $q_d^\varepsilon(u)$ belong to $\cY_d$.
By applying Lemma \ref{ivmqr:lem:uniform-differentiability} to the map $u \mapsto f_{d,z} (q_d^\ast (u)) \det(D q_d^\ast (u)),$ we have

\begin{align*}
    \lim_{\varepsilon \to 0}
    \sup_{u \in A_{d, \varepsilon}}
    \sup_{
        \substack{
            h \in T_{q^\ast}^1
            \\
            q^\ast + \varepsilon h \in \cQ_{0, s}
        }
    }
    \left|
        \frac{
            f_{d, z} (q_d^\varepsilon (u))
            \det D q_d^\varepsilon (u)
            -
            f_{d, z} (q_d^\ast (u))
            \det D q_d^\ast (u)
        }{\varepsilon}
        -
        \div F_{d,z,h} (u)
    \right|
    =
    0
\end{align*}
by Assumption \ref{ivmqr:ass:smooth-densities}, convexity of $\cY_d$, Jacobi's formula, and the Piola identity.
Thus, it follows that
\begin{align*}
    0
    \leq
    \sup_{
        \substack{
            h \in T_{q^\ast}^1
            \\
            q^\ast + \varepsilon h \in \cQ_{0, s}
        }
    }
    I_1 (\varepsilon, h)
    \leq
    |\cD|
    \cdot
    |\cU|
    \cdot
    1
    \cdot
    o (1)
    \to
    0
\end{align*}
as $\varepsilon \to 0.$

By Lemma \ref{ivmqr:lem:support-identification}, $\cY_d = q_d^\ast (\cU) \subset q_d^{\varepsilon, h} (\cU).$
Since $q^{\varepsilon, h} \in \tilde \cQ,$ it follows that for $u, v \in \cU,$
$$
    \norm{q_d^{\varepsilon, h} (u) - q_d^{\varepsilon, h} (v)}
    \geq
    \underline\lambda
    \norm{u - v}
    ,
$$
which implies that $q_d^{\varepsilon, h}$ is injective.
Thus
$$
    T_{d, \varepsilon, h}
    \coloneqq
    (q_d^{\varepsilon, h})^{-1} \circ q_d^\ast
$$
is well-defined and satisfies $T_{d, \varepsilon, h} (\cU) = A_{d, \varepsilon, h} \subset \cU.$
Since $D q_d^{\varepsilon, h}$ is nonsingular, $T_{d, \varepsilon, h}$ has a $C^1$ extension to a neighborhood of $\cU.$

Since $q_d^{\varepsilon, h} (T_{d, \varepsilon, h} (u)) - q_d^{\varepsilon, h} (u) = - \varepsilon h_d (u),$ the bound above implies
$$
    \norm{T_{d, \varepsilon, h} (u) - u}
    \leq
    \frac{1}{\underline \lambda}
    \norm{q_d^{\varepsilon, h} (T_{d, \varepsilon, h} (u)) - q_d^{\varepsilon, h} (u)}
    \leq
    \frac{\varepsilon}{\underline \lambda}
$$
for every $u \in \cU,$ so $\norm{T_{d, \varepsilon, h} - \id}_\infty \leq \varepsilon / \underline \lambda$ holds.
Moreover, since 
$$
    D T_{d, \varepsilon, h} (u)
    =
    \left[
        D q_d^{\varepsilon, h} (T_{d, \varepsilon, h} (u))
    \right]^{-1}
    D q_d^\ast (u)
    ,
$$
the inequality implies 
$$
    \sup_{
        \substack{
            h \in T_{q^\ast}^1
            \\
            q^\ast + \varepsilon h \in \cQ_{0, s}
        }
    }
    \norm{D T_{d, \varepsilon, h} - I}_\infty
    =
    O (\varepsilon)
    .
$$
Similarly, it is not hard to see
\begin{align*}
    \sup_{
        \substack{
            h \in T_{q^\ast}^1
            \\
            q^\ast + \varepsilon h \in \cQ_{0, s}
        }
    }
    \norm{T_{d, \varepsilon, h} - \id + \varepsilon (D q_d^\ast)^{-1} h_d}_\infty
    =
    O(\varepsilon^2)
    .
\end{align*}

Now, we apply Lemma \ref{ivmqr:lem:transport-boundary} with $K = \cU$ and $T = T_{d, \varepsilon, h}.$
Take
$$
    b(u)
    =
    \psi(u)
    f_{d,z}(q_d^\ast(u))
    \det(Dq_d^\ast(u))
    .
$$
This function belongs to $C^1(\cU)$, admits a $C^1$ extension, and satisfies $\norm{b}_{C^1}\leq C\norm{\psi}_{C^1}.$
Using \eqref{ivmqr:eq:transport-remainder}, we obtain
\begin{align*}
    &
    \left|
        \frac{1}{\varepsilon}
        \int_{\cU\setminus A_{d, \varepsilon, h}}
            \psi 
            f_{d,z}(q_d^\ast)
            \det(Dq_d^\ast)
         du
        +
        \int_{\partial\cU}
            \psi 
            f_{d,z}(q_d^\ast)
            \det(Dq_d^\ast)
            \left(
                \frac{T_{d, \varepsilon, h} (u) - u}{\varepsilon}
            \right)^\prime
            \nu
        dS
    \right|
    \\
    &\hspace{8cm}\leq
    \frac{C}{\varepsilon}
    \norm{b}_{C^1}
    \norm{T_{d, \varepsilon, h} - \id}_{C^1}^2
    .
\end{align*}
Using the first-order expansion of $T_{d, \varepsilon, h},$ the second integral on the left-hand side is rewritten as 
$$
    -
    \int_{\partial\cU}
        \psi 
        f_{d,z}(q_d^\ast)
        \det(Dq_d^\ast)
        \left(
        (D q_d^\ast)^{-1} h_d
        \right)^\prime
        \nu
    dS
    +
    o (\varepsilon)
    =
    -
    \int_{\partial\cU}
        \psi 
        F_{d, \varepsilon, h}^\prime
        \nu
    dS
    +
    o (\varepsilon)
    .
$$
Thus, we have
$$
    \left|
        \frac{1}{\varepsilon}
        \int_{\cU\setminus A_{d, \varepsilon, h}}
            \psi 
            f_{d,z}(q_d^\ast)
            \det(Dq_d^\ast)
         du
        -
        \int_{\partial\cU}
            \psi 
            F_{d, \varepsilon, h}^\prime
            \nu
        dS
    \right|
    =
    O (\varepsilon)
    \norm{\psi}_{C^1}
    ,
$$
which implies
$$
    0
    \leq
    \sup_{
        \substack{
            h \in T_{q^\ast}^1
            \\
            q^\ast + \varepsilon h \in \cQ_{0, s}
        }
    }
    I_2 (\varepsilon, h)
    \leq
    |\cD|
    \cdot
    O (\varepsilon)
    \norm{\psi}_{C^1}
    \to
    0
$$
as $\varepsilon \to 0.$

Finally, taking $b = 1$ in \eqref{ivmqr:eq:transport-remainder} gives $|\cU\setminus A_{d,\varepsilon}| = O(\varepsilon),$ which shows 
$$
    0
    \leq
    \sup_{
        \substack{
            h \in T_{q^\ast}^1
            \\
            q^\ast + \varepsilon h \in \cQ_{0, s}
        }
    }
    I_3 (\varepsilon, h)
    \leq
    |\cD|
    \cdot
    O (\varepsilon)
    \to
    0
$$
as $\varepsilon \to 0.$

Now, Assumption \ref{ivmqr:ass:differentiable} has been verified.
\end{proof}

\subsection{Proof of Lemma \ref{ivmqr:lem:argmin-consistency}}

\begin{proof}
It suffices to show that any subsequence of $\hat \theta_n$ contains a further subsequence converging to $\theta_0.$
Since $\Theta$ is compact, we may assume that $\hat \theta_n$ converges to some $\theta_\ast \in \Theta$ up to a subsequence.
By Condition (3), $\theta_\ast \in \Theta_0.$
Since $\hat \theta_n$ is a minimizer of $Q_n,$ we have $Q_n (\hat \theta_n) \leq Q_n (\bar \theta_n)$ where $\bar \theta_n$ is given in Condition (2).
By passing $n \to \infty,$ Condition (4) implies $Q_0 (\theta_\ast) \leq Q_0 (\theta_0).$
By Condition (1), $\theta_\ast = \theta_0.$
\end{proof}

\printbibliography

\end{document}